\documentclass[pdflatex,sn-mathphys-num]{sn-jnl}

\usepackage{xcolor}
\usepackage{graphicx}
\usepackage{multirow}
\usepackage{amsmath,amssymb,amsfonts}
\usepackage{amsthm}
\usepackage{mathrsfs}
\usepackage[title]{appendix}
\usepackage{textcomp}
\usepackage{manyfoot}
\usepackage{booktabs}
\usepackage{algorithm}
\usepackage{algorithmicx}
\usepackage{algpseudocode}
\usepackage{listings}
\usepackage{makecell}
\usepackage{array}
\usepackage{caption}
\theoremstyle{thmstyleone}
\newtheorem{theorem}{Theorem}

\newtheorem{lemma}[theorem]{Lemma}

\theoremstyle{thmstyletwo}

\theoremstyle{thmstylethree}

\newtheorem{mechanism}{Mechanism}
\begin{document}

\title[Article Title]{Mechanism Design for Facility Location Games Under a Prelocated Facility
}

\author[1]{\fnm{Genjie} \sur{Qin}}

\author[1]{\fnm{Qizhi} \sur{Fang}}

\author*[1,2]{\fnm{Wenjing} \sur{Liu}}\email{liuwj@ouc.edu.cn}

\affil*[1]{\orgdiv{School of Mathematical Sciences}, \orgname{Ocean University of China}, \orgaddress{\street{No. 238, Songling Road}, \city{Qingdao}, \postcode{266100}, \state{Shandong Province}, \country{China}}}

\affil[2]{\orgdiv{Laboratory of Marine Mathematics}, \orgname{Ocean University of China}, \orgaddress{\street{No. 238, Songling Road}, \city{Qindao}, \postcode{266100}, \state{Shandong Province}, \country{China}}}


    \abstract{We study the problem of locating a new homogeneous facility under a prelocated facility. Here, a set of $n$ agents is located on a real line or a circle, each of whom has her location as private information, and her cost is the (expected) distance from her location to the nearest facility. Our goal is to design mechanisms which can approximately minimize the maximum cost or the social cost while eliciting agents' private information truthfully (i.e., strategy-proof).
    
    Based on real-life scenarios, we consider the problem in two settings: the general setting where each agent can be located at both sides of the prelocated facility, and the special setting where all the agents are located at the same side of the prelocated facility. For agents on a line, in the general setting, we design the best possible deterministic strategy-proof mechanism with $2$-approximation and provide a lower bound of $1.5-\epsilon\textbf{ }(\epsilon>0)$ for any randomized strategy-proof mechanism under the maximum cost objective. For the social cost, we obtain an upper bound of $n$ for deterministic strategy-proof mechanisms and lower bounds of $1.5$ and $1.0425$ for any deterministic strategy-proof mechanism and any randomized strategy-proof mechanism, respectively. In the special setting, we further provide a randomized strategy-proof $5/3$-approximation mechanism for the maximum cost and a deterministic strategy-proof $(n-1)$-approximation mechanism for the social cost. For agents on a circle, we provide a deterministic strategy-proof 2-approximation mechanism under the maximum cost objective.

    We further introduce a performance measure for mechanisms in facility location problems with prelocated facilities, called the improvement ratio, which is defined as the worst ratio between the improvement achieved by a mechanism and the best possible improvement after adding a new facility. For the maximum cost, we derive tight bounds of 1.5 for deterministic mechanisms in the general setting, and an upper bound of $\frac{15}{13}$ for randomized mechanisms in the special setting. For the social cost, we provide an upper bound of $\frac{n}{2}$ in the special setting.}


\keywords{Facility location game, Approximation, Prelocated facility,  Algorithmic mechanism design}

\maketitle

\section{Introduction}\label{sec1}

Facility location problem is an important problem in the field of optimization. Its main goal is to select the best locations to place facilities under given constraints to optimize certain objectives. Such problems usually involve engineering, operations research, or urban planning, such as medical facility location, logistics center location, or communication base station location.

It is also the prototypical problem used by Procaccia and Tennenholtz  \cite{bib1} when they introduced their highly successful agenda on approximate mechanism design without money in 2009. In the basic setting of this problem, a social planner is tasked with locating a facility or two facilities on a real line based on the location profile reported by agents to minimize the maximum cost or the social cost, and each agent's cost is the distance between her location and the nearest facility. In order to benefit herself, the agent may misreport her location to manipulate the outcome. Consequently, the social planner needs to design a mechanism that outputs an (approximately) optimal outcome while guaranteeing all the agents report truthfully. Since then, this problem has been extensively studied in the literature of theoretical computer science and artificial intelligence.

Most previous work focuses on the locating problem where prelocated facilities are not taken into account. However, in many practical applications, prelocated resources (such as facilities) should not be wasted. For instance, due to population growth, the government needs to build a new primary school to ensure adequate educational coverage in a city. Naturally, the government must determine the location of the new school based on not only where the school-age children live but also the locations of the prelocated primary schools.

In addition, it is also ubiquitous in the real-world that all agents are located at the same side of the prelocated facility. For example, residents in coastal cities usually have a great demand for seafood, but fishermen tend to sell their seafood at the seaside, which makes it inconvenient for the residents who live far from the seaside to buy seafood. Therefore, the government decides to build a new seafood market based on the location information of the residents. For another example, the building manager in a skyscraper usually only sets up life service facilities (such as takeaway cabinets) on the ground floor at the beginning. As the number of residents increases, the demand for facilities increases. To facilitate the residents, the manager considers to add a new service facility on a certain floor based on the locations of the residents who need the service.

In order to capture such scenarios, we study how to locate a new homogeneous facility on the basis of a prelocated facility in the following two settings:
\begin{itemize}
    \item General setting: Each agent can be located at both sides of the prelocated facility.
    \item Special setting: All the agents are located at the same side of the prelocated facility.
\end{itemize}
Our objective is to design mechanisms that can elicit agents' private information truthfully and approximately optimize certain social objectives. In addition, we further introduced the improvement ratio, a performance measure specifically designed for facility location problems with prelocated facilities. It is defined as the worst-case ratio between the improvement achieved by a given mechanism and the optimal improvement attainable by placing a new facility. Even if a mechanism does not achieve the optimal outcome, a low improvement ratio indicates that it makes near-optimal use of the new facility, offering a more practical perspective on mechanism design.

\subsection{Related Work}

Our work is grounded on a string of research for approximate mechanism design without money, which was initiated by Procaccia and Tennenholtz \cite{bib1}. They provided upper and lower bounds on the achievable approximation ratio of strategy-proof mechanisms for one facility and two-facility problems on the real line under two types of social objectives. For the one facility game, they gave the best possible deterministic and randomized mechanisms under each social objective. For the two-facility game under the maximum cost objective, they provided a best possible deterministic mechanism and an upper bound of $5/3$ and a lower bound of $1.5$ for randomized strategy-proof mechanisms. In addition, they gave an upper bound of $n-2$ and a lower bound of $1.5$ for deterministic strategy-proof mechanisms under the social cost objective. Later, Lu et al. \cite{bib3} improved these results by a randomized  strategy-proof $4$-approximation mechanism in general metric spaces and a lower bound of $\frac{n-1}{2}$ for any deterministic strategy-proof mechanism on a line. Fotakis and Tzamos \cite{bib4} raised the lower bound to $n-2$, which ultimately eliminated the gap between
the upper bound and the lower bound on deterministic strategy-proof mechanisms in two-facility games. Alon et al. \cite{bib2} considered the problem of locating a facility on networks. 

Much work further considers the situation where the agents have various preferences for facilities. Cheng et al. \cite{bib21} considered an obnoxious facility game on the network where all agents try to be as far away from the facility as possible. Ye et al. \cite{bib20} studied the problem of locating an obnoxious facility on a real line under the objecives of maximizing the sum of squares of distances (maxSOS) and maximizing the sum of distances (maxSum). Serafino and Ventre \cite{bib18,bib5} considered building two heterogeneous facilities, where the cost of each agent is the sum of the distances between her location and the facilities she is interested in. Yuan et al. \cite{bib6} studied the two-facility games with min-variant and max-variant, where each agent's cost is the distance between her location and the closer facility or the further facility respectively. Li et al. \cite{bib7} improved the upper bound for the min-variant. Anastasiadis and Deligkas \cite{bib8} considered the situation where each agent might want to be close to a facility, be away from a facility, or be indifferent about its presence. Fong et al. \cite{bib9} studied the facility location game with fractional preference where the preference of each agent for facilities is represented by a number between 0 and 1. So far, there are numerous variations of the facility location game, such as other types of cost functions \cite{bib14,bib15}, constraints on facility locations and facilities' capacity \cite{bib12,bib10,bib13,bib11}, etc. Interested readers may refer to a detailed survey by Chan et al. \cite{bib22}.

As far as we know, all the previous work except \cite{bib16,bib17,bib19} does not take prelocated facilities into account. Recently, Chan and Wang \cite{bib16} considered modifying the structure of regions by adding a costless shortcut edge based on a prelocated facility on the real line. Chan et al. \cite{bib17} further studied the problem of adding a non-zero cost shortcut or two costless shortcuts. Qin et al. \cite{bib19} extended the model in \cite{bib16} by building a bridge to connect two separated regions. Compared with \cite{bib16,bib17,bib19}, we study how to locate a new homogeneous facility based on a prelocated facility, which can be considered as another natural perspective of utilizing the prelocated facility.  It is worth mentioning that our model is equivalent to the model in \cite{bib16} when one endpoint of the costless shortcut edge in \cite{bib16} is fixed at the prelocated facility. Therefore, all the mechanisms in our work can be used directly in \cite{bib16} and the mechanism in \cite{bib16} that outputs one endpoint exactly at the prelocated facility can also apply to our model.

\subsection{Our Results}

In this paper, we study the problem of locating a new homogeneous facility based on a prelocated facility on a real line or a circle. We derive upper and lower bounds on the achievable approximation ratio for strategy-proof mechanisms in two settings. The general setting where agents may located at either side of the prelocated facility, and the special setting where all agents are located at one side only. The results are summarized in Table \ref{tab1}, where UB and LB represent the upper bound and the lower bound, respectively.


\begin{table}[h]
\centering
\caption{Bounds on Approximation Ratios On the Line and circle}
\label{tab1}
\begin{tabular*}{\textwidth}{@{\extracolsep\fill}cccccc}
\toprule
\multirow{2.5}{*}{Domain} & \multirow{2.5}{*}{Setting} & \multicolumn{2}{c}{Maximum Cost} & \multicolumn{2}{c}{Social Cost} \\\cmidrule{3-4}\cmidrule{5-6}
& & Deterministic & Randomized & Deterministic & Randomized \\
\midrule
\multirow{2}{*}{Line} & \multirow{2}{*}{General} 
& UB: 2     & UB: 2                 & UB: $n$           & UB: 6 \cite{bib1} \\
                      &        
& LB: 2     & LB: $1.5-\epsilon$    & LB: 1.5           & LB: 1.0425 \\
\midrule
\multirow{2}{*}{Line} & \multirow{2}{*}{Special} 
& UB: 2     & UB: 5/3               & UB: $n-1$         & UB: 6 \cite{bib1} \\
                      &        
& LB: 2     & LB: $1.5-\epsilon$    & LB: 1.0425        & LB: 1.0425 \\
\midrule
\multirow{2}{*}{circle} & \multirow{2}{*}{General} 
& UB: 2     & UB: 2                 & UB: $n$           & UB: 6 \cite{bib1} \\
                       &        
& LB: 2     & LB: $1.5-\epsilon$    & LB: 1.5           & LB: 1.0425 \\
\botrule
\end{tabular*}
\end{table}
We also introduce the improvement ratio, a performance measure specially for mechanisms in facility location problems with prelocated facilities. We establish theoretical bounds on the improvement ratio of strategy-proof mechanisms under both the maximum and social cost objectives, in both the general and special settings. The results are summarized in Table \ref{tab2}.
\begin{table}[htp]\label{tab2}
\centering
\caption{\centering Bounds on Improvement Rations On the Line}\label{tab2}
\begin{tabular*}{\textwidth}{@{\extracolsep\fill}ccccc}
\toprule%
& \multicolumn{2}{@{}c@{}}{General Setting} & \multicolumn{2}{@{}c@{}}{Special Setting} \\\cmidrule{2-3}\cmidrule{4-5}%
Social Objective & Deterministic & Randomized & Deterministic & Randomized  \\
\midrule
Maximum & UB: 1.5 & UB: 1.5 & UB: 1.5 & UB: 15/13 \\
cost & LB: 1.5 & LB: 1.01 & LB: 1.5 &  LB: 1.01\\
\midrule
Social cost & \textbackslash  & \textbackslash  & UB: n/2  &\textbackslash \\

\botrule
\end{tabular*}
\end{table}

\noindent\textbf{Remarks.} In Table \ref{tab1}, under the social cost objective, we only obtain a non-constant upper bound for deterministic strategy-proof mechanisms in both settings. However, Mechanism 3 in \cite{bib16} provides an upper bound of 6 for randomized strategy-proof mechanisms in our model since this mechanism always locates one endpoint of the costless shortcut at the prelocated facility. On the other hand, Mechanism \ref{mec1} in our work can be used in \cite{bib16} to improve their upper bound from 3 to 2 under the maximum cost objective. Further, the proof of Theorem \ref{thm9} in this paper can also be used in \cite{bib16} to improve their lower bound of any randomized strategy-proof mechanism from 1.02 to 1.0425 under the social cost objective.

\section{Model}\label{sec2}

Suppose a facility has already been located on the real line, and its location is publicly known. Without loss of generality, assume that the facility is located at $y_{0}=0$. There is a set of $n$ agents \(N=\{1,2,...,n\}\) who need to receive service from the facility. For each agent $i\in N$, she has a private location $x_{i}\in \mathbb{R}$ and let $\mathbf{x}=(x_1,x_2,\ldots,x_n)$ be a location profile or an instance. We want to construct a new homogeneous facility on the real line and every agent can be served by the nearest one.

A deterministic mechanism is a function $f : \mathbb{R}^{n}\rightarrow\mathbb{R}$ which maps a location profile to a facility location. If $f( \mathbf{x})=y$, the cost of agent $i\in N$ is $cost(x_{i},f( \mathbf{x}))=\min\left\{|x_{i}|,|x_{i}-y|\right\}$. A randomized mechanism is a function $f$ which maps a location profile to a probability distribution over facility locations. Formally, \( f \) can be defined as \( f: \mathbb{R}^n \rightarrow \Delta(\mathbb{R}) \), where \( \Delta(\mathbb{R}) \) represents the set of probability distributions over \( \mathbb{R} \) and the cost of agent $i\in N$ is $cost(x_{i},f( \mathbf{x}))=\mathbb{E}_{Y\sim f( \mathbf{x})}\min\left\{|x_{i}|,|x_{i}-Y|\right\}$.

A mechanism is strategy-proof if no agent can benefit from misreporting her location regardless of the locations reported by the others. Formally, $\forall \mathbf{x}\in \mathbb{R}^{n}, \forall i\in N, \forall x_{i}'\in\mathbb{R}, cost(x_{i},f( \mathbf{x}))\leq cost\left(x_{i},f(x_{i}', \mathbf{x}_{-i})\right).$ Here, $ \mathbf{x}_{-i}$ represents the location vector of the set of agents  $N\!\setminus\!\{i\}$.

The maximum cost and social cost of a mechanism $f$ with respect to a location profile \( \mathbf{x}\) are defined as the maximum cost among all $n$ agents and the total cost of all $n$ agents, respectively. For deterministic mechanisms $f$,
$
MC( \mathbf{x},f)=\max_{i\in N}cost(x_{i},f( \mathbf{x})) $ and $ SC( \mathbf{x},f)=\sum_{i\in N}cost(x_{i},f( \mathbf{x})).
$
While for randomized mechanisms $f$, $MC( \mathbf{x},f)=\mathbb{E}_{Y\sim f( \mathbf{x})}[MC(\textbf{x},Y)]$ and $SC( \mathbf{x},f)=\mathbb{E}_{Y\sim f( \mathbf{x})}[SC(\textbf{x},Y)].$

Let the optimal solution for an instance $ \mathbf{x}$ under the maximum cost objective or the social cost objective be denoted as $OPT_{MC}(\mathbf{x})$ and $OPT_{SC}(\mathbf{x})$, respectively. The subscript is omitted without confusion. A strategy-proof mechanism \( f \) has an \textbf{approximation ratio} of $\gamma\textbf{ }(\geq 1)$ under the maximum cost objective, if
\begin{equation*}
\gamma=\sup_{\mathbf{x}\in\mathbb{R}^{n}}\frac{MC(\mathbf{x},f)}{MC(\mathbf{x},OPT_{MC})}.
\end{equation*}
The approximation ratio is defined similarly under the social cost objective. 

In this paper, we will focus on anonymous\footnote{A mechanism is anonymous if its outcome depends only on the agents' locations, not on their identities.} strategy-proof mechanisms with good approximation ratios under the social objectives of minimizing the maximum cost and the social cost. Henceforth, for simplicity, the instance $\mathbf{x}$ is assumed to satisfy $x_{1}\leq x_{2}\leq...\leq x_{n}$ without further comment.\\
\textbf{Notations.} For \( \forall \mathbf{x} \in \mathbb{R}^{n} \), \( x_{1} \) and \( x_{n} \) are the two endpoints. Define \( L(\mathbf{x}) \) and \( S(\mathbf{x}) \) as the locations of the endpoints that are farther and closer to $y_{0}$, respectively. Define
\begin{equation*}
    l( \mathbf{x})=\left\{
\begin{array}{cc}
  \min\left\{x_{i}|x_{i}>\frac{L( \mathbf{x})}{3}\right\}, & \text{if } L( \mathbf{x})\geq0, \\
  \max\left\{x_{i}|x_{i}<\frac{L( \mathbf{x})}{3}\right\}, & \text{if } L( \mathbf{x})<0.
\end{array}
\right.
\end{equation*}
\begin{equation*}
    b( \mathbf{x})=\left\{
\begin{array}{cc}
  \max\left\{x_{i}|x_{i}\leq\frac{L( \mathbf{x})}{3}\right\}, & \text{if } L( \mathbf{x})\geq0, \\
  \min\left\{x_{i}|x_{i}\geq\frac{L( \mathbf{x})}{3}\right\}, & \text{if } L( \mathbf{x})<0.
\end{array}
\right.
\end{equation*}

\section{Line}
In this section, we study the approximation performance of strategy-proof mechanisms for maximum cost and the social cost objectives. We focus on the real line, analyzing both general and special settings.

\subsection{General Setting}\label{sec3}
This section discusses how to locate a new (homogeneous) facility on the assumption that agents may be located on either side of the prelocated facility. We will study deterministic and randomized strategy-proof mechanisms under the objectives of minimizing the maximum cost or the social cost respectively.

\subsubsection{Maximum Cost}\label{subsec2}

For the maximum cost, we first present the optimal solution and the optimal value. Then we design a deterministic strategy-proof $2$-approximation mechanism which is also the best possible deterministic mechanism. We also prove a lower bound of $1.5-\epsilon$ for any randomized strategy-proof mechanism.

\begin{theorem}\label{thm1}
For $\forall  \mathbf{x} \in \mathbb{R}^{n}$, the optimal solution for minimizing the maximum cost is given by
$$OPT( \mathbf{x})=\frac{l( \mathbf{x})+L( \mathbf{x})}{2}.$$
The optimal value is
$$
MC( \mathbf{x},OPT)=\left\{
\begin{array}{cc}
  \max\left\{|S( \mathbf{x})|,|b( \mathbf{x})|,\frac{|L( \mathbf{x})-l( \mathbf{x})|}{2}\right\}, & \text{if }S( \mathbf{x})\cdot L( \mathbf{x})<0, \\
  \max\left\{|b( \mathbf{x})|,\frac{\left|L( \mathbf{x})-l( \mathbf{x})\right|}{2}\right\}, & \text{if }S( \mathbf{x})\cdot L( \mathbf{x})\geq0.
\end{array}
\right.
$$
\end{theorem}
\begin{proof}
Consider the case of $S(\mathbf{x})\cdot L(\mathbf{x})<0$. Without loss of generality, let $S(\mathbf{x})\leq0\leq L(\mathbf{x})$. Obviously,  $MC(\mathbf{x},OPT)\geq|S(\mathbf{x})|$. If $OPT(\mathbf{x})\geq\frac{2L(\mathbf{x})}{3}$, then the cost of agent on $b(\mathbf{x})$ is $b(\mathbf{x})$. Furthermore, the maximum cost of the agents on $l(\mathbf{x})$ and $L(\mathbf{x})$ is at least $\frac{L(\mathbf{x})-l(\mathbf{x})}{2}$. Combining the above, we have $MC(\mathbf{x},OPT)\geq\max\left\{|S( \mathbf{x})|,|b( \mathbf{x})|,\frac{|L(\mathbf{x})-l( \mathbf{x})|}{2}\right\}$. It is easy to see that $\frac{l(\mathbf{x})+L(\mathbf{x})}{2}$ achieves a maximum cost of at most $\max\left\{|S(\mathbf{x})|,|b(\mathbf{x}),\frac{|L(\mathbf{x})-l(\mathbf{x}))|}{2}|\right\}$, which indicates the optimality. The proof for the case of $S(\mathbf{x})\cdot L(\mathbf{x})\geq0$ is similar.
\end{proof}
Note that the optimal mechanism is not strategy-proof. For example, consider an instance $\mathbf{x}=(x_{1},x_{2})=(-3,4)$. Obviously, $OPT(\mathbf{x})=4,cost(x_{1},OPT(\mathbf{x}))=3$. However, if agent 1 misreports her location as $-5$, denoting $\mathbf{x'}=(-5,4)$, we have $cost(x_{1},OPT(\mathbf{x'}))=cost(x_{1},-5)=2$. Thus, agent 1 can decrease her cost by misreporting.

\begin{theorem}\label{thm2}
    Any deterministic strategy-proof mechanism has an approximation ratio of at least $2$ for the maximum cost.
\end{theorem}
\begin{proof}
    Assume there exists a deterministic strategy-proof mechanism $ f $ with approximation ratio less than 2. Consider an instance $  \mathbf{x}=(x_{1},x_{2})=(1,2) $. It holds that $ MC( \mathbf{x},f) < 2 \cdot MC( \mathbf{x},OPT) = 1 $. Therefore, $ f( \mathbf{x}) \in (1,2) $.  Without loss of generality, let $ f(\mathbf{x}) = 1 + \epsilon $ where $ \epsilon \in (0, \frac{1}{2}] $. Consider the instance $  \mathbf{x'}=(x_{1},x_{2}')=(1,1+\epsilon) $, then $ f( \mathbf{x'}) \in (1, 1+\epsilon) $. Agent 2 can benefit by reporting $ x_{2}'$ to $x_{2}$, which contradicts the strategy-proofness of $f$.
\end{proof}

\begin{mechanism}\label{mec1}
    Given $\mathbf{x}=(x_{1},...,x_{n})\in\mathbb{R}^{n}$, if $|x_{n}|\geq|x_{1}|$, the facility is located at $y$ as follows:
$$
   y=\left\{
    \begin{array}{cl}
    x_{n} &, \text{if } 0 \leq x_{1} \leq x_{n}, \\
    \max\left\{2|x_{1}|, x_{n}\right\} &, \text{if } x_{1} < 0 \leq x_{n}.
    \end{array}
    \right.
$$
If $|x_{n}|<|x_{1}|$, $y$ is defined symmetrically.
\end{mechanism}
\begin{theorem}\label{thm3}
Mechanism \ref{mec1} is strategy-proof.
\end{theorem}
\begin{proof} Denote Mechanism \ref{mec1} by \( f \). Given any instance \(  \mathbf{x}=(x_{1},x_{2},...,x_{n})\in\mathbb{R}^{n} \), we need to prove that every agent $i\in N$ cannot benefit by misreporting her location $x_{i}$ as $x_{i}'\in \mathbb{R}$. Denote $\mathbf{x'}=(x_{i}',\mathbf{x}_{-i})$. The proof falls into the following cases:

\textbf{Case 1} \( 0\leq x_{1}\leq x_{n}. \)

Obviously, \(f( \mathbf{x})=x_{n} \) and any agent $i$ with $x_{i}=x_{n}$ has no incentive to lie. Now consider agent $i$ with $x_{i}<x_{n}$. If $x_{i}'>x_{i}$, then $f(\mathbf{x'})\geq f(\mathbf{x})$, which implies that the facility  is moving further away from agent $i$. If $x_{i}'<x_{i}$, then $f(\mathbf{x'})=f(\mathbf{x})$ when $x_{i}'\geq -\frac{x_{n}}{2}$, $f(\mathbf{x'})=2|x_{i}'|\geq f(\mathbf{x})$ when $-x_{n}\leq x_{i}'<-\frac{x_{n}}{2}$ and $f(\mathbf{x'})\leq x_{i}'<-x_{n}$ when $x_{i}'<-x_{n}$, none of which can make agent $i$'s cost decrease.

\textbf{Case 2} \( x_{1}<0\leq x_{n} \) and \( 2|x_{1}|\leq x_{n} \).

In this case, $f(\mathbf{x})=x_{n}$. For agent $i$ with $x_{i}\geq0$, the proof is similar to Case 1. Now consider agent $i$ with $x_{i}<0$. Here, $cost(x_{i},f(\mathbf{x}))=|x_{i}|\leq x_{1}\leq\frac{x_{n}}{2}$. Note that $f(\mathbf{x'})\geq f(\mathbf{x})$ when $x_{i}'>-x_{n}$ and $f(\mathbf{x'})\leq x_{i}'$ when $x_{i}'<-x_{n}$. Both of them keep agent $i$'s cost unchanged.

\textbf{Case 3} \( x_{1}<0\leq x_{n} \) and \( 2|x_{1}|> x_{n} \).

In this case, $f(\mathbf{x})=2|x_{1}|$. For agent $i$ with $x_{i}\geq0$, her cost will never decrease by misreporting. For agent $i$ with $x_{i}<0$, $cost(x_{i},f(\mathbf{x}))=|x_{i}|\leq|x_{1}|\leq x_{n}.$ If $x_{i}'\geq -x_{n}$, then $f(\mathbf{x'})\geq x_{n}\geq0$. Otherwise, $f(\mathbf{x'})=-\max\left\{2x_{n},|x_{i}'|\right\}\leq-2x_{n}$, which implies $|x_{i}-f(\mathbf{x'})|\geq x_{n}\geq |x_{i}|$. In both cases, the cost of agent $i$ have not changed.
\end{proof}

\begin{theorem}\label{thm4}
Mechanism \ref{mec1} is $2$-approximation under the maximum cost objective.
\end{theorem}

\begin{proof}
Denote Mechanism \ref{mec1} by $f$. Without loss of generality, consider any instance \(  \mathbf{x}=(x_{1},x_{2},...,x_{n})\) with \(|x_{n}|\geq|x_{1}| \).

If $0\leq x_{1}\leq x_{n}$, then $f( \mathbf{x})=x_{n}$,
  \begin{align*}
    \frac{MC( \mathbf{x},f)}{MC( \mathbf{x},OPT)} &\leq \frac{\max\left\{b( \mathbf{x}),x_{n}-l( \mathbf{x})\right\}}{\max\left\{b( \mathbf{x}),\frac{x_{n}-l( \mathbf{x})}{2}\right\}}\leq2.
  \end{align*}

If $x_{1}<0\leq x_{n},2|x_{1}|\leq x_{n}$, then $f( \mathbf{x})=x_{n}$,
  \begin{align*}
    \frac{MC( \mathbf{x},f)}{MC( \mathbf{x},OPT)} &\leq\frac{\max\left\{|x_{1}|,b( \mathbf{x}),x_{n}-l( \mathbf{x})\right\}}{\max\left\{|x_{1}|,b( \mathbf{x}),\frac{x_{n}-l( \mathbf{x})}{2}\right\}}\leq2.
  \end{align*}

If $x_{1}<0\leq x_{n},2|x_{1}|> x_{n}$, then $f( \mathbf{x})=2x_{1}$,
  \begin{align*}
    \frac{MC( \mathbf{x},f)}{MC( \mathbf{x},OPT)} &\leq\frac{ \max\left\{|x_{1}|,\frac{x_{n}-l( \mathbf{x})}{2}\right\}}{\max\left\{|x_{1}|,|x_{n}|\right\}}=\frac{|x_{1}|}{\max\left\{|x_{1}|,|x_{n}|\right\}}=1.
  \end{align*}
The proof is completed.
\end{proof}

\begin{theorem}\label{thm5}
    Any randomized strategy-proof mechanism has an approximation ratio of at least $1.5-\epsilon$ for the maximum cost, $\epsilon>0$.
\end{theorem}
\begin{proof}
   Let $f$ be any randomized strategy-proof mechanism. Let $M>0$ be sufficiently large. Consider an instance $  \mathbf{x}=(x_{1},x_{2})=(M+1,M+2)$. Obviously, $\sum_{i\in\{1,2\}} cost(x_{i},f(\mathbf{x}))=E_{Y\sim f( \mathbf{x})}\left[cost(x_{1},Y)+cost(x_{2},Y)\right]\geq1$. Without loss of generality, assume that  $cost(x_{1},f( \mathbf{x}))\geq\frac{1}{2}$.

   Now consider another instance $ \mathbf{x'}=(x_{1}',x_{2}')=(M,M+2)$. Note that $OPT( \mathbf{x'})=M+1$ and $MC( \mathbf{x'},OPT)=1$. By strategy-proofness, we have $cost(M+1,f( \mathbf{x'}))=cost(x_{1},f(\mathbf{x'}))\geq cost(x_{1},f( \mathbf{x}))\geq \frac{1}{2}$. Otherwise, the agent located at $x_{1}$ in $\mathbf{x}$ can benefit by misreporting her location as $x_{1}'=M$. The maxmum cost of $f$ w.r.t. $\mathbf{x'}$ is
\begin{align*}
    MC( \mathbf{x'},f) =& E_{Y^{'}\sim f( \mathbf{x}')}\left[\max_{i\in\{1,2\}}cost(x_{i}',Y')\right]\\
    =&Pr\{Y'\leq2M\}\cdot E_{Y^{'}\sim f( \mathbf{x}')}\left[\max_{i\in\{1,2\}}cost(x_{i}',Y')|Y'\leq2M\right]\\
    &+Pr\{Y'>2M\}\cdot E_{Y^{'}\sim f( \mathbf{x}')}\left[\max_{i\in\{1,2\}}cost(x_{i}',Y')|Y'>2M\right].
\end{align*}

Note that $\max_{i\in\{1,2\}}cost(x_{i}',Y')=1+cost(M+1,Y')$ when $Y'\leq2M$ and $\max_{i\in\{1,2\}}cost(x_{i}',Y')>M\geq cost(M+1,Y')-1$ when $Y'>M$. We analyze $MC(\mathbf{x'},f)$ according to the following cases.

\textbf{Case 1} $Pr\{Y'>2M\}\geq\frac{3}{2M}$.
$$MC( \mathbf{x'},f)\geq Pr\{Y'>2M\}\cdot E_{Y^{'}\sim f( \mathbf{x}')}\left[\max_{i\in\{1,2\}}cost(x_{i}',Y')|Y'>2M\right]\geq\frac{3}{2M}\cdot M=\frac{3}{2}.$$

\textbf{Case 2} $Pr\{Y'<2M\}\geq\frac{3}{2M}$.
\begin{align*}
    MC( \mathbf{x'},f) \geq&Pr\{Y'\leq2M\}\cdot E_{Y^{'}\sim f( \mathbf{x}')}\left[1+cost(M+1,Y')|Y'\leq2M\right] \\
    &+Pr\{Y'>2M\}\cdot E_{Y^{'}\sim f( \mathbf{x}')}\left[1+cost(M+1,Y')-2|Y'>2M\right]\\
    =& 1+cost(M+1,Y')-2Pr\{Y'>2M\}\\
    \geq& \frac{3}{2}-\frac{3}{M}.
\end{align*}

Therefore,
$$\frac{MC( \mathbf{x'},f)}{MC( \mathbf{x'},OPT)}\geq\frac{3}{2}-\frac{3}{M}.$$
\end{proof}

\subsubsection{Social cost}\label{subsec2}

We prove that Mechanism \ref{mec1} is $n$-approximation for the social cost. We obtain lower bounds of $1.5$ and $1.0425$ for any deterministic and randomized strategy-proof mechanism, respectively. 

\begin{lemma}\label{lem6}
Denote Mechanism \ref{mec1} by $f$. For any instance $\mathbf{x}=(x_{1},...,x_{n})$ with $0\leq x_{1}\leq x_{n}$,
$$SC(\mathbf{x},f)\leq(n-1)\cdot SC(\mathbf{x},OPT).$$
\end{lemma}
\begin{figure}[htbp]
    \centering
    \includegraphics[width=1\linewidth]{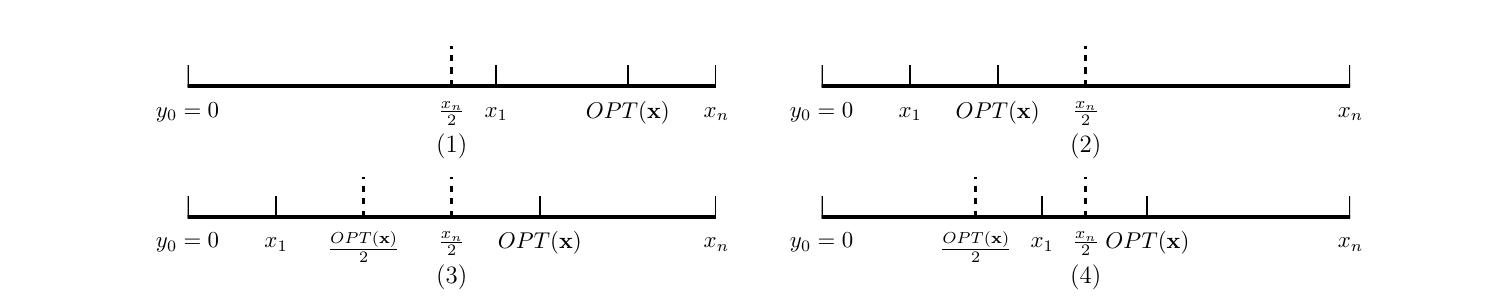}
    \caption{An instance $\mathbf{x}=(x_{1},x_{2},...,x_{n})$ with $0\leq x_{1}\leq x_{n}$}
    \label{fig1}
\end{figure}

\begin{proof}

Note that $x_{1}\leq OPT(\mathbf{x})\leq x_{n}$ and $f(\mathbf{x})=x_{n}$. The proof is performed according to the following cases:

\textbf{Case 1} $x_{1}\geq\frac{x_{n}}{2}$, as shown in Fig. \ref{fig1}(1).

It is clear that $OPT( \mathbf{x}) = x_{\lfloor{\frac{n+1}{2}}\rfloor}$ and $cost(x_{i},f(\mathbf{x}))=x_{n}-x_{i},\forall i\in N.$ We have
\begin{align*}
  SC( \mathbf{x},OPT) &= \sum_{i\in N}\left|x_{i}-x_{\lfloor{\frac{n+1}{2}}\rfloor}\right|\geq|x_{1}-x_{\lfloor{\frac{n+1}{2}}\rfloor}|+|x_{n}-x_{\lfloor{\frac{n+1}{2}}\rfloor}|=x_{n}-x_{1},\\
    SC( \mathbf{x},f) &= \sum_{i\in N}cost(x_{i},f(\mathbf{x}))\leq \sum_{i=1}^{n-1}\left(x_{n}-x_{1}\right)\leq (n-1)\cdot SC(\mathbf{x},OPT).
  \end{align*}

\textbf{Case 2} $x_{1}<\frac{x_{n}}{2}$.

\textbf{Case 2.1} $OPT( \mathbf{x}) \leq \frac{x_{n}}{2}$, as shown in Fig. \ref{fig1}(2).

\begin{align*}
   SC( \mathbf{x},OPT) &\geq cost(x_{n},OPT(\mathbf{x}))=x_{n}-OPT(\mathbf{x})\geq\frac{x_{n}}{2},\\
    SC( \mathbf{x},f) &= \sum_{i\in N}cost(x_{i},f(\mathbf{x}))=\sum_{i:x_{i}\leq\frac{x_{n}}{2}}x_{i}+\sum_{i:x_{i}>\frac{x_{n}}{2}}(x_{n}-x_{i})\leq (n-1)\cdot SC(\mathbf{x},OPT).
\end{align*}

\textbf{Case 2.2} $OPT( \mathbf{x}) > \frac{x_{n}}{2}$ and $x_{1} < \frac{OPT(\mathbf{x})}{2}$, as shown in Fig. \ref{fig1}(3).

Partition $N$ as follows:
\begin{align*}
N_{1}&=\left\{i\in N\Big|x_{1}\leq x_{i}\leq\frac{1}{2}OPT(\mathbf{x})\right\},\text{ } N_{2}=\left\{i\in N\Big|\frac{1}{2}OPT(\mathbf{x})< x_{i}\leq\frac{x_{n}}{2}\right\},\\
N_{3}&=\left\{i\in N\Big|\frac{x_{n}}{2}< x_{i}\leq OPT(\mathbf{x})\right\},\text{ } N_{4}=\left\{i\in N\Big|OPT(\mathbf{x})< x_{i}\leq x_{n}\right\}.
\end{align*}

Thus, we have
\begin{align*}
  SC( \mathbf{x},OPT)\geq&\sum_{i\in N_{1}\cup N_{2}\cup N_{3}}cost(x_{i},f(\mathbf{x}))+\left(x_{n}-OPT(\mathbf{x})\right),\\
    SC(\mathbf{x},f) =&\sum_{i\in N_{1}}cost(x_{i},OPT(\mathbf{x}))+\sum_{i\in N_{2}}cost(x_{i},f(\mathbf{x}))\\
    &+\sum_{i\in N_{3}}\left[cost(x_{i},OPT(\mathbf{x}))+x_{n}-OPT(\mathbf{x})\right]+\sum_{i\in N_{4}}cost(x_{i},f(\mathbf{x}))\\
    \leq& \sum_{i\in N_{1}}cost(x_{i},OPT(\mathbf{x}))+\sum_{i\in N_{2}}\left[\frac{x_{n}}{2}+cost(x_{i},OPT(\mathbf{x}))-OPT(\mathbf{x})+\frac{x_{n}}{2}\right]\\
    &+\sum_{i\in N_{3}}\left[cost(x_{i},OPT(\mathbf{x}))+x_{n}-OPT(\mathbf{x})\right]+\sum_{i\in N_{4}-\{x_{n}\}}\left[x_{n}-OPT(\mathbf{x})\right]\\
    \leq& \sum_{i\in N_{1}\cup N_{2}\cup N_{3}}cost(x_{i},OPT(\mathbf{x}))+|N_{2}\cup N_{3}\cup N_{4}-\{x_{n}\}|\cdot(x_{n}-OPT(\mathbf{x}))\\
    \leq& SC(\mathbf{x},OPT)+(n-2)\cdot(x_{n}-OPT(\mathbf{x})),\\
    \frac{SC(\mathbf{x},f)}{SC(\mathbf{x},OPT)}\leq&\frac{SC(\mathbf{x},OPT)+(n-2)\cdot(x_{n}-OPT(\mathbf{x}))}{SC(\mathbf{x},OPT)}\leq1+\frac{(n-2)\cdot(x_{n}-OPT(\mathbf{x}))}{x_{n}-OPT(\mathbf{x})}=n-1.
\end{align*}

\textbf{Case 2.3} $OPT( \mathbf{x}) > \frac{x_{n}}{2}$ and $x_{1} > \frac{OPT( \mathbf{x})}{2}$, as shown in Fig. \ref{fig1}(4).

$$SC(\mathbf{x},OPT)\geq cost(x_{1},OPT(\mathbf{x}))+cost(x_{n},OPT(\mathbf{x}))=x_{n}-x_{1}>\frac{x_{n}}{2}.$$
It is trivial that
$$SC(\mathbf{x},f)\leq(n-1)\cdot SC(\mathbf{x},OPT).$$
\end{proof}

\begin{theorem}\label{thm7}
Mechanism \ref{mec1} is $n$-approximation under the social cost objective.
\end{theorem}
\begin{figure}[htbp]
    \centering
    \includegraphics[width=1\linewidth]{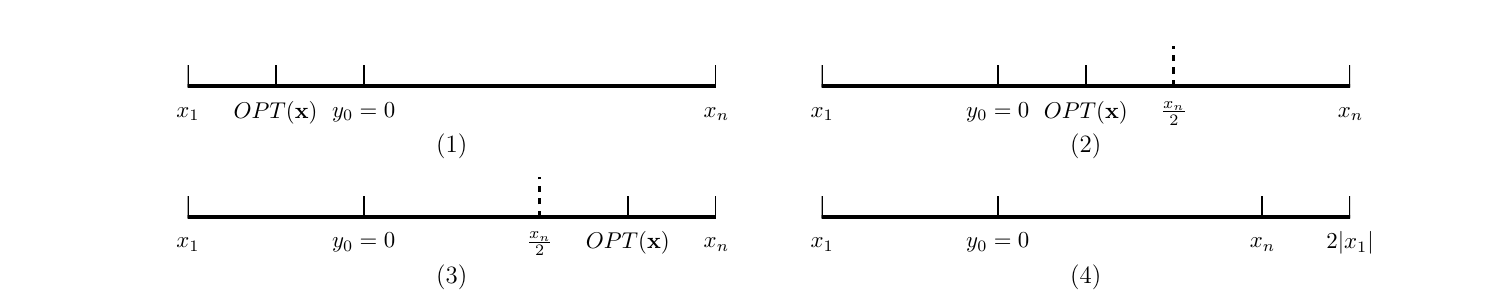}
    \caption{An instance $\mathbf{x}=(x_{1},x_{2},...,x_{n})$ with $x_{1}<0\leq x_{n}$}
    \label{fig2}
\end{figure}

\begin{proof}
 Without loss of generality, given any instance \(  \mathbf{x}=(x_{1},x_{2},...,x_{n})\) with \(|x_{n}|\geq|x_{1}| \). Trivially, $x_{1}\leq OPT(\mathbf{x})\leq x_{n}$. Denote Mechanism \ref{mec1} by $f$, and our analysis is as follows.

\textbf{Case 1} $0\leq x_{1}\leq x_{n}$. By Lemma \ref{lem6}, we have
$$
SC( \mathbf{x},f)\leq (n-1)\cdot SC( \mathbf{x},OPT).
$$

\textbf{Case 2} $x_{1}<0\leq x_{n}$ and $|2x_{1}|\leq x_{n}$. We have $f(\mathbf{x})=x_{n}$ and
\begin{equation}
    SC(\mathbf{x},f)=\sum_{i\in N}cost(x_{i},f(\mathbf{x}))=\sum_{i:x_{i}\leq\frac{x_{n}}{2}}|x_{i}|+\sum_{i:x_{i}>\frac{x_{n}}{2}}(x_{n}-x_{i}). \tag{$*$}
\end{equation}

\textbf{Case 2.1} $OPT( \mathbf{x})\leq 0$, as shown in Fig. \ref{fig2}(1).

Note that $SC(\mathbf{x},OPT)\geq cost(x_{n},OPT(\mathbf{x}))= x_{n}$. Thus, by $(*)$, we have
$$
SC(\mathbf{x},f)\leq (n-1)\cdot \frac{x_{n}}{2}\leq \frac{n-1}{2}\cdot SC(\mathbf{x},OPT).
$$

\textbf{Case 2.2} $0<OPT(\mathbf{x})\leq\frac{x_{n}}{2}$, as shown in Fig. \ref{fig2}(2).

Since $SC(\mathbf{x},OPT)\geq cost(x_{n},OPT(\mathbf{x}))\geq\frac{x_{n}}{2}$, then by $(*)$ ,
$$SC(\mathbf{x},f)\leq (n-1)\cdot SC(\mathbf{x},OPT).$$

\textbf{Case 2.3} $OPT( \mathbf{x})>\frac{x_{n}}{2}$, as shown in Fig. \ref{fig2}(3).

By a similar proof as Case 2.2 in Lemma \ref{lem6}, we can obtain $$SC(\mathbf{x},f)\leq (n-1)\cdot SC(\mathbf{x},OPT).$$

\textbf{Case 3} $x_{1}<0\leq x_{n},|2x_{1}|> x_{n}$, as shown in Fig. \ref{fig2}(4).

In this case, $f(\mathbf{x})=2|x_{1}|.$ Note that $SC(\mathbf{x},OPT)\geq|x_{1}|$. We have
$$
    SC(\mathbf{x},f)=\sum_{i\in N}cost(x_{i},f(\mathbf{x}))=\sum_{i:x_{i}\leq|x_{1}|}|x_{i}|+\sum_{i:x_{i}>x_{1}}(2x_{1}-x_{i})\leq n\cdot|x_{1}|\leq n\cdot SC(\mathbf{x},OPT).
$$
\textbf{Tight Example.} Let $ \mathbf{x}$ be an instance with $n$ agents, where $n-1$ agents are located at $1$ and one agent is located at $-1$. It is clear that $OPT( \mathbf{x})=1,SC(\mathbf{x},OPT)=1$. Since $f(\mathbf{x})=2,SC(\mathbf{x},f)=n-1+1=n$, we have $SC(\mathbf{x},f)=n\cdot SC(\mathbf{x},OPT)$.
\end{proof}

\begin{theorem}\label{thm8}
    Any deterministic strategy-proof mechanism has an approximation ratio of at least $1.5$ for the social cost.
\end{theorem}
\begin{proof}
 Assume there exists a deterministic strategy-proof mechanism $ f $ with approximation ratio less than $1.5$. Consider an instance $  \mathbf{x}=(x_{1},x_{2})=(-1,1) $. Without loss of generality, let $f( \mathbf{x})=y\leq0$, we have $cost(x_{2},y)=1$. Consider the instance $  \mathbf{x'}=(x_{1},x_{2}')=(-1,1.5) $ and $f( \mathbf{x'})=y'$, we have $ SC( \mathbf{x'},y') < 1.5 \cdot SC( \mathbf{x'},OPT) = 1.5$. It means $y'>0$, we have $cost(x_{2}',y')=SC( \mathbf{x'},y')-cost(x_{1},y')<\frac{1}{2}$. It can lead to the $cost(x_{2},y')<1$. Agent 2 can benefit by reporting $ x_{2}$ to $x_{2}'$, which contradicts with the strategy-proofness of $f$.
\end{proof}
Before our proof of the lower bound for any randomized strategy-proof mechanism, recall that strategy-proofness is equivalent to partial group strategy-proofness for facility location games where each agent has her location as private information  \cite{bib3}. In other words, for any group of agents located at the same location, no member can benefit if they misreport simultaneously.

\begin{theorem}\label{thm9}
    Any randomized strategy-proof mechanism has an approximation ratio of at least $1.0425$ for the social cost.
\end{theorem}
\begin{proof}
    Assume there exists a randomized strategy-proof mechanism $f$ with an approximation ratio less than 1.0425. Let $ \mathbf{x}$ be an instance with 7 agents, where 4 agents are located at $x_{1}=0.7$ and 3 agents are located at $x_{2}=2$. It is clear that $OPT( \mathbf{x})=2$ and $SC( \mathbf{x},OPT)=2.8$. Assume the mechanism $f$ outputs $Y \geq 1.4$ with probability $p$. Since $E_{Y\sim f(\mathbf{x})} \left[SC( \mathbf{x},Y)|Y<1.4\right]\geq SC( \mathbf{x},0.7)=3.9$,
\begin{align*}
SC( \mathbf{x},f)&\geq p\cdot2.8+(1-p)\cdot3.9=3.9-1.1\cdot p,\\
  \frac{SC(\mathbf{x},f)}{SC( \mathbf{x},OPT)} &< 1.0425\Rightarrow3.9-1.1\cdot p<1.0425\times2.8 \Rightarrow p>\frac{981}{1100}.
\end{align*}
Then the cost of agent 1 is
$$cost(x_{1},f( \mathbf{x})) > p\cdot 0.7+(1-p)\cdot 0>\frac{6867}{11000}>0.62427. $$

Now consider another instance $\mathbf{x'}$ with 7 agents, where 4 agents are located at $x_{1}'=1$ and 3 agents are located at $x_{2}=2$.
It is clear that $OPT( \mathbf{x'})=1$ and $SC( \mathbf{x'},OPT)=3$.
Suppose
the mechanism $f$ outputs $Y' \in [0.17, 1.23]$ with probability $q$ for instance $ \mathbf{x'}$. Since $E_{Y^{'}\sim f(\mathbf{x'})}\left[SC( \mathbf{x'}, Y')|Y' \notin [0.17, 1.23]\right] \geq SC( \mathbf{x'}, 1.23) = 4 \times 0.23 + 3 \times 0.77 = 3.23$,
\[
SC( \mathbf{x'}, f) \geq q \cdot 3 + (1-q) \cdot 3.23 = 3.23 - 0.23 \cdot q
\]

Given that the approximation ratio of any randomized strategy-proof mechanism is less than $1.0425$,
\[
\frac{SC( \mathbf{x'}, f)}{SC( \mathbf{x'}, OPT)} < 1.0425 \Rightarrow 3.23 - 0.23 \cdot q < 1.0425 \times 3 \Rightarrow q > \frac{41}{92}.
\]
\[
cost(x_{1}, f( \mathbf{x'})) < q \cdot 0.53 + (1-q) \cdot 0.7 < \frac{5743}{9200} < 0.62427.
\]
Thus,
\[
cost(x_{1}, f( \mathbf{x'})) < cost(x_{1}, f( \mathbf{x})).
\]
It contradicts the strategy-proofness of the mechanism.
\end{proof}

\subsubsection{Discussion}

For the social cost, there is a huge gap between the upper bound of $n$ and the lower bound of 1.5 for deterministic strategy-proof  mechanisms. We conjecture the deterministic lower bound is $\Omega(n)$ just as in two-facility location games \cite{bib4}, but failed to verify it.

It is noteworthy that Mechanism 3 in \cite{bib16} outputs one endpoint of the shortcut edge at point 0 and the other at $x_{k}$ with probability $\frac{|x_{k}|}{\sum_{i\in N}|x_{i}|}, k\in N$, which is proved to be strategy-proof and 6-approximation under the social cost objective.
This can be described in our model as follows.

\begin{mechanism}\label{mec2}\rm{\cite{bib16}}
 For $\forall \mathbf{x}$, for any agent $k$, the mechanism outputs $x_{k}$ with probability $\frac{|x_{k}|}{\sum_{i\in N}|x_{i}|}$.
\end{mechanism}

\begin{theorem}\label{thm10}\rm{\cite{bib16}}
Mechanism \ref{mec2} is a randomized group strategy-proof $6$-approximation mechanism for the social cost.
\end{theorem}

In addition, the instances and methods in the proof of Theorem \ref{thm9} can be used in \cite{bib16} to improve the lower bound of any randomized strategy-proof mechanism under the social cost objective from 1.02 to 1.0425.

\subsection{Special Setting}\label{sec4}

This section considers the scenario where all agents are located on the same side of the prelocated facility. Without loss of generality, assume all of them are on the right side of $y_0=0$, i.e., any instance $\mathbf{x}=(x_{1},x_{2},...,x_{n})$ satisfies $0\leq x_{1}\leq x_{2}\leq ...\leq x_{n}$. Then $L(\mathbf{x})=x_{n},S(\mathbf{x})=x_{1}$, the optimal solution for minimizing the maximum cost is $OPT(\mathbf{x})=\frac{l(\mathbf{x})+L(\mathbf{x})}{2}$, and the optimal value is
$$MC(\mathbf{x},OPT)=\max\left\{b(\mathbf{x}),\frac{|L(\mathbf{x})-l(\mathbf{x})|}{2}\right\}.$$

In this section, the impossibility results in Section \ref{sec3} still hold, except for that of the deterministic mechanism for the social cost. Mechanism \ref{mec1} which outputs $L(\mathbf{x})=x_n$ for any instance $\mathbf{x}$ in this setting, is still 2-approximation and the best deterministic strategy-proof mechanism for the maximum cost.
Furthermore, we provide a randomized strategy-proof 5/3-approximation mechanism for the maximum cost.

\begin{mechanism}\label{mec3}
For an instance $ \mathbf{x}=(x_{1},\ldots,x_{n})$, the facility is located at $y$, where $y$ is defined as a random point according to the following cases:

\textbf{Case 1} $b( \mathbf{x}) \geq L( \mathbf{x}) - l( \mathbf{x})$. Let
$$
y = \left\{
\begin{array}{cl}
L(\mathbf{x}) - b(\mathbf{x}) & , \text{with probability (w.p.) } \frac{1}{6}, \\[5pt]
\frac{2L(\mathbf{x}) - b(\mathbf{x})}{2} & , \text{w.p. } \frac{1}{3}, \\[5pt]
L(\mathbf{x}) & , \text{w.p. } \frac{1}{2}.
\end{array}
\right.
$$

\textbf{Case 2} $b( \mathbf{x})< L( \mathbf{x})-l( \mathbf{x})$. Let
$$
    y=\left\{
    \begin{array}{cl}
      \max\left\{l( \mathbf{x}),\frac{2L( \mathbf{x})}{3}\right\} & ,\text{with probability (w.p.) }\frac{1}{6}, \\[5pt]
      \frac{\max\left\{l( \mathbf{x}),\frac{2L( \mathbf{x})}{3}\right\}+L( \mathbf{x})}{2} & ,\text{w.p. }\frac{1}{3}, \\[5pt]
      L( \mathbf{x}) &, \text{w.p. }\frac{1}{2}.
    \end{array}
    \right.
$$
\end{mechanism}

\begin{theorem}\label{thm11}
Mechanism \ref{mec3} is a randomized strategy-proof $\frac{5}{3}$-approximation mechanism under the maximum cost objective.
\end{theorem}
\begin{figure}[htbp]
    \centering
    \includegraphics[width=1\linewidth]{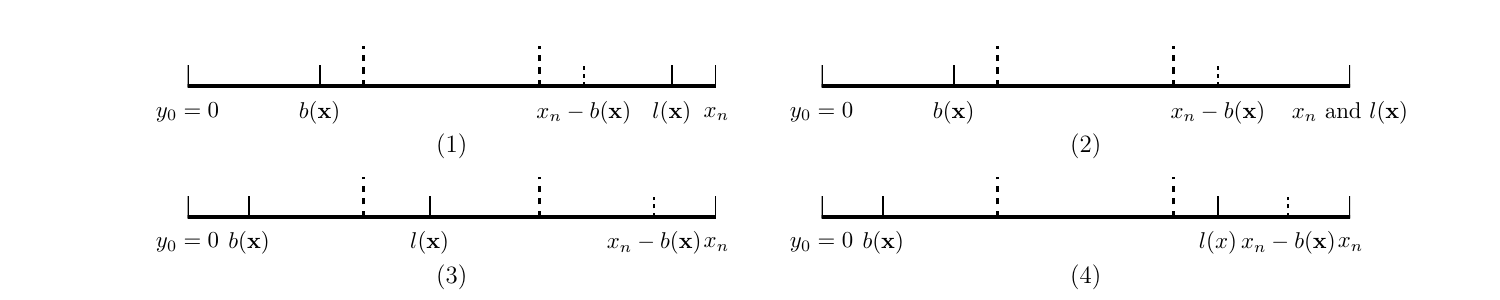}
    \caption{An instance $\mathbf{x}=(x_{1},x_{2},...,x_{n})$ with $0\leq x_{1} \leq x_{n}$}
    \label{fig3}
\end{figure}

\begin{proof}
Denote Mechanism \ref{mec3} by $f$.

\noindent\textbf{Strategy-proofness.}
 We need to prove that for any instance $ \mathbf{x}$, every agent $i\in N$ cannot benefit by misreporting her location $x_{i}$ as $x_{i}'\in \mathbb{R}$. Denote $\mathbf{x'}=(x_{i}',\mathbf{x}_{-i})$ and the possible locations of \( f(\mathbf{x}) \) as $A, \frac{A+B}{2}, B$ from left to right. Our analysis falls into the following cases:

\textbf{Case 1} $b( \mathbf{x}) \geq L( \mathbf{x}) - l( \mathbf{x})$ and $l( \mathbf{x}) \neq L( \mathbf{x})$, as shown in Fig. \ref{fig3}(1). Suppose agent $i$ reports her location $x_{i}$ to $x_{i}'$.

\textbf{Case 1.1} $x_{i} \leq b( \mathbf{x})$. Agent $i$ always choose the prelocated facility, then we have $cost(x_{i},f(\mathbf{x}))\leq cost(x_{i},f(x_{i}',\mathbf{x}_{-i}))$.

\textbf{Case 1.2} $x_{i} \in [l( \mathbf{x}), L( \mathbf{x}))$.

\textbf{Case 1.2.1} $x_{i}' \leq L( \mathbf{x})$. The possible locations of $f(x_{i}',  \mathbf{x}_{-i})$ are denoted as $A', \frac{A'+B'}{2}, B'$ from left to right. We have $A' \in \left[\frac{2 \cdot L( \mathbf{x})}{3}, A\right]$. Let $\Delta = A - A' \geq 0$, then
\begin{align*}
cost(x_{i}, f( \mathbf{x})) &= \frac{1}{6} \cdot (x_{i} - A) + \frac{1}{3} \left|x_{i} - \frac{A + B}{2}\right| + \frac{1}{2} \cdot (B - x_{i}), \\
cost(x_{i}, f(x_{i}',  \mathbf{x}_{-i})) &= \frac{1}{6} \cdot \left|x_{i} - A'\right| + \frac{1}{3} \left|x_{i} - \frac{A' + B}{2}\right| + \frac{1}{2} \cdot (B - x_{i}) \\
&\geq \frac{1}{6} \cdot (x_{i} - A + \Delta) + \frac{1}{3} \left[\left|x_{i} - \frac{A + B}{2}\right| - \frac{\Delta}{2}\right] + \frac{1}{2} \cdot (B - x_{i}) \\
&= cost(x_{i}, f( \mathbf{x})).
\end{align*}

\textbf{Case 1.2.2} $x_{i}' > L( \mathbf{x})$. Let $\Delta = x_{i}' - L( \mathbf{x})$. We have $L( \mathbf{x'}) = x_{i}'$. It holds that $cost(x_{i},f(\mathbf{x}))\leq cost(x_{i},f(x_{i}',\mathbf{x}_{-i}))$ when $b( \mathbf{x}) \neq b( \mathbf{x'})$. Now consider the case of $b( \mathbf{x}) = b( \mathbf{x'})$. Let $\Delta_{1} = b( \mathbf{x}) - [L( \mathbf{x}) - l( \mathbf{x'})]$ and $\Delta_{2} = \frac{3}{2} \cdot l( \mathbf{x'}) - L( \mathbf{x})$, we have
\begin{equation*}
    L( \mathbf{x}) + \Delta_{1} - l( \mathbf{x'}) = b( \mathbf{x'}) = b( \mathbf{x}), L( \mathbf{x}) + \Delta_{2} - l( \mathbf{x'}) = \frac{1}{3} \cdot [L( \mathbf{x}) + \Delta_{2}].
\end{equation*}

(1) $\Delta \leq \Delta_{1}$. The possible locations of $f(x_{i}',  \mathbf{x}_{-i})$ are denoted as $A_{1}, \frac{A_{1} + B_{1}}{2}, B_{1}$ from left to right. It is evident that $b( \mathbf{x'}) \geq L( \mathbf{x}) + \Delta - l( \mathbf{x'})$, we have
$$ A_{1} = A + \Delta, B_{1} = B + \Delta. $$
\begin{align*}
  cost(x_{i},f( \mathbf{x})) &= \frac{1}{6}\cdot(x_{i}-A)+\frac{1}{3}\left|x_{i}-\frac{A+B}{2}\right|+\frac{1}{2}\cdot (B-x_{i}), \\
  cost(x_{i},f(x_{i}', \mathbf{x}_{-i})) &= \frac{1}{6}\cdot\left|x_{i}-A_{1}\right|+\frac{1}{3}\left|x_{i}-\frac{A_{1}+B_{1}}{2}\right|+\frac{1}{2}\cdot(B_{1}-x_{i}) \\
  &\geq \frac{1}{6}\cdot(x_{i}-A-\Delta)+\frac{1}{3}\left[\left|x_{i}-\frac{A+B}{2}\right|-\Delta\right]+\frac{1}{2}\cdot(\Delta+B-x_{i}) \\
   &= cost(x_{i},f( \mathbf{x}))-\frac{\Delta}{6}-\frac{\Delta}{3}+\frac{\Delta}{2}\\
   &= cost(x_{i},f( \mathbf{x})).
\end{align*}

(2) \( \Delta \in (\Delta_{1}, \Delta_{2}) \). The possible locations of $f(x_{i}',  \mathbf{x}_{-i})$ are denoted as \( A_{2}, \frac{A_{2}+B_{2}}{2}, B_{2} \) from left to right. It is known that \( b( \mathbf{x'}) < L( \mathbf{x}) + \Delta - l( \mathbf{x'}) < \frac{1}{3} \cdot [L( \mathbf{x}) + \Delta] \),
$$ A_{2} = l( \mathbf{x'}) = A + \Delta_{1}, B_{2} = L( \mathbf{x}) + \Delta = B + \Delta. $$
\begin{align*}
  cost(x_{i},f( \mathbf{x})) &= \frac{1}{6}\cdot(x_{i}-A)+\frac{1}{3}\left|x_{i}-\frac{A+B}{2}\right|+\frac{1}{2}\cdot (B-x_{i}) ,\\
  cost(x_{i},f(x_{i}', \mathbf{x}_{-i})) &= \frac{1}{6}\cdot\left|x_{i}-A_{2}\right|+\frac{1}{3}\left|x_{i}-\frac{A_{2}+B_{2}}{2}\right|+\frac{1}{2}\cdot(B_{2}-x_{i}) \\
  &\geq \frac{1}{6}\cdot(x_{i}-A-\Delta_{1})+\frac{1}{3}\left[\left|x_{i}-\frac{A+B}{2}\right|-\frac{\Delta_{1}+\Delta}{2}\right]+\frac{1}{2}\cdot(\Delta+B-x_{i}) \\
   &= cost(x_{i},f( \mathbf{x}))-\frac{\Delta_{1}}{6}-\frac{\Delta_{1}}{6}-\frac{\Delta}{6}+\frac{\Delta}{2}\\
   &= cost(x_{i},f( \mathbf{x}))-\frac{\Delta_{1}}{3}+\frac{\Delta}{3}\\
   &\geq cost(x_{i},f( \mathbf{x})).
\end{align*}

(3) \( \Delta \geq \Delta_{2} \). The possible locations of $f(x_{i}',  \mathbf{x}_{-i})$ are denoted as \( A_{3}, \frac{A_{3}+B_{3}}{2}, B_{3} \) from left to right. It is known that \( b( \mathbf{x'}) < \frac{1}{3} \cdot [L( \mathbf{x}) + \Delta] \leq L( \mathbf{x}) + \Delta - l( \mathbf{x'}) \), we have
$$ A_{3} = \frac{2}{3} \cdot [L( \mathbf{x}) + \Delta] = A + \Delta_{1} + \frac{2}{3} \cdot (\Delta - \Delta_{2}), B_{3} = B + \Delta. $$
\begin{align*}
  cost(x_{i},f( \mathbf{x})) &= \frac{1}{6}\cdot(x_{i}-A)+\frac{1}{3}\left|x_{i}-\frac{A+B}{2}\right|+\frac{1}{2}\cdot (B-x_{i}), \\
  cost(x_{i},f(x_{i}', \mathbf{x}_{-i})) &= \frac{1}{6}\cdot\left|x_{i}-A_{3}\right|+\frac{1}{3}\left|x_{i}-\frac{A_{3}+B_{3}}{2}\right|+\frac{1}{2}\cdot(B_{3}-x_{i}) \\
  &\geq \frac{1}{6}\cdot\left\{x_{i}-A-\left[\Delta_{1}+\frac{2\cdot(\Delta-\Delta_{2})}{3}\right]\right\}\\
  &+\frac{1}{3}\left[\left|x_{i}-\frac{A+B}{2}\right|-\frac{\Delta+\Delta_{1}+\frac{2}{3}\cdot(\Delta-\Delta_{2})}{2}\right]+\frac{1}{2}\cdot(\Delta+B-x_{i}) \\
   &= cost(x_{i},f( \mathbf{x}))-\frac{\Delta_{1}}{6}-\frac{\Delta-\Delta_{2}}{9}-\frac{1}{3}\cdot\frac{\Delta+\Delta_{1}+\frac{2}{3}\cdot(\Delta-\Delta_{2})}{2}+\frac{\Delta}{2}\\
   &\geq cost(x_{i},f( \mathbf{x}))+\frac{\Delta}{9}+\frac{2}{9}\cdot\Delta_{2}-\frac{\Delta_{1}}{3}\\
   &\geq cost(x_{i},f( \mathbf{x})).
\end{align*}

\textbf{Case 1.3} $x_{i}=L( \mathbf{x})$. If $x_{i}'\leq L( \mathbf{x})$, the facility will move away from $L( \mathbf{x})$, then $cost(x_{i},f(\mathbf{x}))\leq cost(x_{i},f(x_{i}',\mathbf{x}_{-i}))$. If $x_{i}'> L( \mathbf{x})$, the proof is similar to that of Case 1.2.2.

\textbf{Case 2} $b( \mathbf{x})\geq L( \mathbf{x})-l( \mathbf{x})$ and $l( \mathbf{x})=L( \mathbf{x})$, as shown in Fig. \ref{fig3}(2). Suppose agent $i$ reports her location $x_{i}$ to $x_{i}'$.

\textbf{Case 2.1} $x_{i}\leq b( \mathbf{x})$. It is obvious that $cost(x_{i},f(\mathbf{x}))\leq cost(x_{i},f(x_{i}',\mathbf{x}_{-i}))$.

\textbf{Case 2.2} $x_{i}=L( \mathbf{x})$.

(1) Only agent $ i $ is located at $ L( \mathbf{x}) $.

If $ x_{i}' \leq L( \mathbf{x}) $, $cost(x_{i},f(\mathbf{x}))\leq cost(x_{i},f(x_{i}',\mathbf{x}_{-i}))$. If $ x_{i}' > L( \mathbf{x}) $, let $ \Delta = x_{i}' - L( \mathbf{x}) > 0 $, we have
\begin{align*}
  cost(x_{i},f( \mathbf{x})) &= \frac{1}{6}\cdot b( \mathbf{x})+\frac{1}{3}\cdot\frac{b( \mathbf{x})}{2}+\frac{1}{2}\cdot0,\\
 cost(x_{i},f(x_{i}', \mathbf{x}_{-i})) &= \frac{1}{6}\cdot \left|b( \mathbf{x})-\Delta\right|+\frac{1}{3}\cdot\left|\frac{b( \mathbf{x})}{2}-\Delta\right|+\frac{\Delta}{2} \\
  &\geq cost(x_{i},f(\mathbf{x}))-\frac{\Delta}{6}-\frac{\Delta}{3}+\frac{\Delta}{2}\\
   &= cost(x_{i},f( \mathbf{x})).
\end{align*}

(2) There are multiple agents at location $L( \mathbf{x})$.

If $x_{i}' \leq L( \mathbf{x})$, The result can only be that the facility is farther from $L( \mathbf{x})$. If $x_{i}' > L( \mathbf{x})$, the proof is similar to Case 1.2.2 shows agents located at $L( \mathbf{x})$ will not profit.

\textbf{Case 3} $b( \mathbf{x}) < L( \mathbf{x}) - l( \mathbf{x})$ and $l( \mathbf{x}) \leq \frac{2}{3} \cdot L( \mathbf{x})$, as shown in Fig. \ref{fig3}(3). Suppose agent $i$ reports her location $x_{i}$ to $x_{i}'$.

\textbf{Case 3.1} $x_{i} \leq b( \mathbf{x})$. Obviously, $cost(x_{i},f(\mathbf{x}))\leq cost(x_{i},f(x_{i}',\mathbf{x}_{-i}))$.

\textbf{Case 3.2} $x_{i} \in \left[l( \mathbf{x}), L( \mathbf{x})\right]$. If $x_{i}' \leq L( \mathbf{x})$, $cost(x_{i},f(\mathbf{x}))\leq cost(x_{i},f(x_{i}',\mathbf{x}_{-i}))$. If $x_{i}' > L( \mathbf{x})$, then $L( \mathbf{x'}) = x_{i}'$. Let $\Delta = L( \mathbf{x'}) - L( \mathbf{x}) > 0$. As $\Delta$ increases from zero, let $\Delta$ be $\Delta_{1}$ when $b( \mathbf{x'}) \neq b( \mathbf{x})$ is satisfied for the first time. It is characterized by:
$$\frac{1}{3} \cdot \left[L( \mathbf{x}) + \Delta_{1}\right] = b( \mathbf{x'}).$$
Let $\Delta_{2}$ be the value of $\Delta$ when $b( \mathbf{x'}) \geq L( \mathbf{x'}) - l( \mathbf{x'})$ is satisfied for the first time after $\Delta \geq \Delta_{1}$. Let $\Delta_{3}$ be the value of $\Delta$ when $b( \mathbf{x'}) \leq L( \mathbf{x'}) - l( \mathbf{x'})$ is satisfied for the first time after $\Delta \geq \Delta_{2}$. Let $\Delta_{4}$ be the value of $\Delta$ when $L( \mathbf{x'}) - l( \mathbf{x'}) = \frac{L( \mathbf{x'})}{3}$ is satisfied for the first time after $\Delta \geq \Delta_{3}$.

(1) $\Delta < \Delta_{1}$. If $x_{i} = l( \mathbf{x})$, it is obvious the possible facility is farther from $\frac{2 \cdot L( \mathbf{x})}{2}$, agents will not lie. If $x_{i} \neq l( \mathbf{x})$, the possible locations of $f(x_{i}',  \mathbf{x}_{-i})$ are denoted as \( A_{1}, \frac{A_{1}+B_{1}}{2}, B_{1} \) from left to right. It is evident that $b( \mathbf{x'}) = b( \mathbf{x}) < x_{i}' - l( \mathbf{x'})$, and we have
$$A_{1} = A + \frac{2}{3} \cdot \Delta, B_{1} = B + \Delta.$$
\begin{align*}
  cost(x_{i},f( \mathbf{x})) &= \frac{1}{6}\cdot \left|x_{i}-A\right|+\frac{1}{3}\cdot\left|x_{i}-\frac{A+B}{2}\right|+\frac{1}{2}\cdot(B-x_{i}),\\
  cost(x_{i},f(x_{i}', \mathbf{x}_{-i})) &= \frac{1}{6}\cdot \left|x_{i}-A_{1}\right|+\frac{1}{3}\cdot\left|x_{i}-\frac{A_{1}+B_{1}}{2}\right|+\frac{1}{2}\cdot(B_{1}-x_{i})\\
  &= \frac{1}{6}\cdot \left|x_{i}-\left(A+\frac{2}{3}\cdot\Delta\right)\right|+\frac{1}{3}\cdot\left|x_{i}-\left(\frac{A+B}{2}+\frac{5}{6}\cdot\Delta\right)\right|+\frac{1}{2}\cdot(B+\Delta-x_{i}) \\
  &\geq cost(x_{i},f( \mathbf{x}))-\frac{1}{9}\cdot\Delta-\frac{5}{18}\cdot\Delta+\frac{1}{2}\cdot\Delta\\
   &\geq cost(x_{i},f( \mathbf{x})).
\end{align*}

(2) $\Delta \in [\Delta_{1}, \Delta_{2})$. The possible locations of $f(x_{i}',  \mathbf{x}_{-i})$ are denoted as \( A_{2}, \frac{A_{2}+B_{2}}{2}, B_{2} \) from left to right. It is evident that $b( \mathbf{x'}) \leq L( \mathbf{x'}) - l( \mathbf{x'})$, and we have
$$A_{2} = A + \frac{2}{3} \cdot \Delta, B_{2} = B + \Delta.$$
According to the proof in (1), her cost will not decrease.

(3) $\Delta \in [\Delta_{2}, \Delta_{3})$, where $b( \mathbf{x'}) \geq L( \mathbf{x'}) - l( \mathbf{x'})$. Let $\delta_{1} = \Delta - \Delta_{2} \geq 0$ and the possible locations of $f(x_{i}',  \mathbf{x}_{-i})$ are denoted as \( A_{3}, \frac{A_{3}+B_{3}}{2}, B_{1} \) from left to right. We have
$$A_{3} = A + \frac{2}{3} \cdot \Delta_{2} + \delta_{1}, B_{3} = B + \Delta_{2} + \delta_{1}.$$
\begin{align*}
cost(x_{i}, f( \mathbf{x})) &= \frac{1}{6} \cdot \left|x_{i} - A\right| + \frac{1}{3} \cdot \left|x_{i} - \frac{A + B}{2}\right| + \frac{1}{2} \cdot (B - x_{i}), \\
cost(x_{i}, f(x_{i}',  \mathbf{x}_{-i})) &= \frac{1}{6} \cdot \left|x_{i} - A_{3}\right| + \frac{1}{3} \cdot \left|x_{i} - \frac{A_{3} + B_{3}}{2}\right| + \frac{1}{2} \cdot (B_{3} - x_{i}) \\
&= \frac{1}{6} \cdot \left|x_{i} - \left(A + \frac{2}{3} \cdot \Delta_{2} + \delta_{1}\right)\right| + \frac{1}{3} \cdot \left|x_{i} - \left(\frac{A + B}{2} + \frac{5}{6} \cdot \Delta_{2} + \delta_{1}\right)\right| \\
&\quad + \frac{1}{2} \cdot (B + \Delta_{2} + \delta_{1} - x_{i}) \\
&\geq cost(x_{i}, f( \mathbf{x})) - \frac{1}{9} \cdot \Delta_{2} - \frac{5}{18} \cdot \Delta_{2} + \frac{1}{2} \cdot \Delta_{2} - \frac{\delta_{1}}{6} - \frac{\delta_{1}}{3} + \frac{\delta_{1}}{2} \\
&\geq cost(x_{i}, f( \mathbf{x})).
\end{align*}

(4) $\Delta \in [\Delta_{3}, \Delta_{4})$, where $b( \mathbf{x'}) \leq L( \mathbf{x'}) - l( \mathbf{x'})$ and $l( \mathbf{x'}) \geq \frac{2 \cdot L( \mathbf{x'})}{3}$. Let $\delta_{2} = \Delta - \Delta_{3} \geq 0$ and The possible locations of $f(x_{i}',  \mathbf{x}_{-i})$ are denoted as \( A_{4}, \frac{A_{4}+B_{4}}{2}, B_{4} \) from left to right. We have
$$A_{4} = A + \frac{2}{3} \cdot \Delta_{2} + \Delta_{3} - \Delta_{2} = A - \frac{\Delta_{2}}{3} + \Delta_{3}, B_{3} = B + \Delta_{3} + \delta_{2}.$$
\begin{align*}
  cost(x_{i},f( \mathbf{x})) &= \frac{1}{6}\cdot \left|x_{i}-A\right|+\frac{1}{3}\cdot\left|x_{i}-\frac{A+B}{2}\right|+\frac{1}{2}\cdot(B-x_{i}),\\
  cost(x_{i},f(x_{i}', \mathbf{x}_{-i})) &= \frac{1}{6}\cdot \left|x_{i}-A_{4}\right|+\frac{1}{3}\cdot\left|x_{i}-\frac{A_{4}+B_{4}}{2}\right|+\frac{1}{2}\cdot(B_{4}-x_{i})\\
  &= \frac{1}{6}\cdot \left|x_{i}-\left(A-\frac{\Delta_{2}}{3}+\Delta_{3}\right)\right|+\frac{1}{3}\cdot\left|x_{i}-\left(\frac{A+B}{2}+\Delta_{3}-\frac{2\cdot\Delta_{2}}{3}+\frac{\delta_{2}}{2}\right)\right|\\
  &\quad+\frac{1}{2}\cdot(B+\Delta_{3}+\delta_{2}-x_{i})\\
  &\geq cost(x_{i},f( \mathbf{x}))+\frac{5\cdot\Delta_{2}}{18}+\frac{\delta_{2}}{3}\\
   &\geq cost(x_{i},f( \mathbf{x})).
\end{align*}

(5) $\Delta \geq \Delta_{4}$. As $\Delta$ increases gradually from $\Delta_{4}$, $L( \mathbf{x'}) - l( \mathbf{x'}) > \frac{L( \mathbf{x'})}{3}$ is always satisfied before $b( \mathbf{x'})$ changes. The possible locations of $f(x_{i}',  \mathbf{x}_{-i})$ are denoted as \( A_{5}, \frac{A_{5}+B_{5}}{2}, B_{5} \) from left to right. We have
$$A_{5} = A + \frac{2}{3} \cdot \Delta, B_{5} = B + \Delta.$$
The proof is similar to Case 3.2(1).

When $\Delta$ continues to increase until $b( \mathbf{x'})$ changes for the first time, the output facility location is far enough away from $x_{i}$, and the cost of agent $i$ will not be reduced.

\textbf{Case 4} $b( \mathbf{x}) < L( \mathbf{x}) - l( \mathbf{x})$ and $l( \mathbf{x}) > \frac{2}{3} \cdot L( \mathbf{x})$. If $l( \mathbf{x}) = L( \mathbf{x})$. It is easy to prove any agents can not benefit by misreporting their location; otherwise, as shown in Fig. \ref{fig3}(4). Suppose agent $i$ reports her location $x_{i}$ to $x_{i}'$.

\textbf{Case 4.1} $x_{i} \leq b( \mathbf{x})$. Obviously, $cost(x_{i},f(\mathbf{x}))\leq cost(x_{i},f(x_{i}',\mathbf{x}_{-i}))$.

\textbf{Case 4.2} $x_{i} \in \big[l( \mathbf{x}), L( \mathbf{x})\big)$. If $x_{i}' \leq L( \mathbf{x})$. The proof is similar to Case 1.2.1. If $x_{i}' = L( \mathbf{x}) + \Delta > L( \mathbf{x})$, it is clear that $cost(x_{i},f(\mathbf{x}))\leq cost(x_{i},f(x_{i}',\mathbf{x}_{-i}))$ when $b( \mathbf{x}) \neq b( \mathbf{x'})$. If $b( \mathbf{x}) = b( \mathbf{x'})$, as $\Delta$ increases from 0, the proof is similar to Case 1.2.2.

\textbf{Case 4.3} $x_{i} = L( \mathbf{x})$. If $x_{i}'\leq L( \mathbf{x})$, the facility will move away from $L( \mathbf{x})$ and agent $i$  can not benefit. If $x_{i}'> L( \mathbf{x})$, the proof is similar to the second part of Case 4.2.

In conclusion, Mechanism \ref{mec3} is a strategy-proof mechanism.

\noindent\textbf{Approximation Ratio.}
Given any instance $\mathbf{x}=(x_{1},x_{2},...,x_{n})$, our analysis proceeds as follows.

\textbf{Case 1}  $b( \mathbf{x}) \geq L( \mathbf{x}) - l( \mathbf{x})$, as shown in Fig. \ref{fig3}(1) or Fig. \ref{fig3}(2).

In this case, we have  $MC(\mathbf{x},OPT) = b( \mathbf{x})$. For each output of $f$, the maximum cost is less than $b(\mathbf{x})$, then it holds that
$$\frac{MC( \mathbf{x},f)}{MC( \mathbf{x},OPT)} \leq \frac{b( \mathbf{x})}{MC( \mathbf{x},OPT)} = 1.$$

\textbf{Case 2} $b( \mathbf{x}) < L( \mathbf{x}) - l( \mathbf{x})$ and $l( \mathbf{x}) \leq \frac{2}{3} \cdot L( \mathbf{x})$.

Let $l( \mathbf{x}) = \frac{L(\mathbf{x})}{3}+ \Delta \leq \frac{2L(\mathbf{x})}{3}, \Delta \in (0, \frac{L(\mathbf{x})}{3}]$. We have
\begin{align*}
  MC( \mathbf{x},OPT) &= \max\left\{b( \mathbf{x}),\frac{L( \mathbf{x})-l( \mathbf{x})}{2}\right\}=\max\left\{b( \mathbf{x}),\frac{L(\mathbf{x})}{3}-\frac{\Delta}{2}\right\}. \\
  MC( \mathbf{x},f) &\leq \frac{1}{6}\cdot \frac{L( \mathbf{x})}{3}+\frac{1}{3}\cdot\max\left\{b( \mathbf{x}),cost(l( \mathbf{x}),\frac{5L(\mathbf{x})}{6} )\right\}+\frac{1}{2}\cdot (L( \mathbf{x})-l( \mathbf{x}))\\
  &= \frac{L(\mathbf{x})}{18}+\frac{1}{3}\cdot\max\left\{b( \mathbf{x}),\frac{L( \mathbf{x})}{2}-\Delta\right\}+\frac{L(\mathbf{x})}{3}-\frac{\Delta}{2}.
\end{align*}

\textbf{Case 2.1} $b( \mathbf{x}) \geq \frac{L(\mathbf{x})}{2} - \Delta \in \left[\frac{L(\mathbf{x})}{6}, \frac{L(\mathbf{x})}{2}\right)$.
\begin{align*}
\frac{MC( \mathbf{x},f)}{MC( \mathbf{x},OPT)} &= \frac{\frac{L(\mathbf{x})}{18} + \frac{b( \mathbf{x})}{3} + \frac{L(\mathbf{x})}{3}- \frac{\Delta}{2}}{b( \mathbf{x})} \leq \frac{1}{3} + \frac{\frac{7L(\mathbf{x})}{18} + \frac{b( \mathbf{x})}{2} - \frac{L(\mathbf{x})}{4}}{b( \mathbf{x})} \\
&= \frac{1}{3} + \frac{14L(\mathbf{x}) + 18b( \mathbf{x}) - 9L(\mathbf{x})}{36 b( \mathbf{x})} \leq \frac{1}{3} + \frac{1}{2} + \frac{5}{6} = \frac{5}{3}.
\end{align*}

\textbf{Case 2.2} $b( \mathbf{x}) < \frac{L(\mathbf{x})}{2} - \Delta$.
$$\frac{MC( \mathbf{x},f)}{MC( \mathbf{x},OPT)} \leq \frac{\frac{L(\mathbf{x})}{18} + \frac{1}{3} \cdot \left(\frac{L(\mathbf{x})}{2} - \Delta\right) + \frac{L(\mathbf{x})}{3}- \frac{\Delta}{2}}{\frac{L(\mathbf{x})}{3} - \frac{\Delta}{2}} = \frac{10L(\mathbf{x}) - 15\Delta}{6L(\mathbf{x}) - 9\Delta} = \frac{5}{3}.$$

\textbf{Case 3} $b( \mathbf{x}) < L( \mathbf{x}) - l( \mathbf{x})$ and $l( \mathbf{x}) > \frac{2}{3} \cdot L( \mathbf{x})$. We have
\begin{align*}
MC( \mathbf{x},OPT) &= \max\left\{b( \mathbf{x}), \frac{L( \mathbf{x}) - l( \mathbf{x})}{2}\right\}, \\
MC( \mathbf{x},f) &= \frac{1}{6} \cdot \max\left\{b( \mathbf{x}), L( \mathbf{x}) - l( \mathbf{x})\right\} + \frac{1}{3} \cdot MC( \mathbf{x},OPT) \\
&\quad+ \frac{1}{2} \cdot \max\left\{b( \mathbf{x}), L( \mathbf{x}) - l( \mathbf{x})\right\}.\\
\frac{MC( \mathbf{x},f)}{MC( \mathbf{x},OPT)} &\leq \frac{2}{6} + \frac{1}{3} + \frac{1}{2} = \frac{5}{3}.
\end{align*}
\textbf{Tight Example.} Consider an instance $\mathbf{x}=(4,6)$. It is clear that $OPT(\mathbf{x})=5$ and the optimal value is 1. Mechanism \ref{mec3} outputs $4$, $5$ and $6$ with probabilities $\frac{1}{6}$, $\frac{1}{3}$ and $\frac{1}{2}$, respectively. Therefore,
\begin{equation*}
        \frac{MC(\mathbf{x},f)}{MC(\mathbf{x},OPT)}=\frac{5}{3}.
\end{equation*}
\end{proof}

For the social cost objective, according to Lemma \ref{lem6}, we have the following theorem.
\begin{theorem}\label{thm13}
Mechanism \ref{mec1} is a deterministic strategy-proof $(n-1)$-approximation mechanism under the social cost objective.
\end{theorem}

\section{Circle}
In this section, we study the approximation performance of strategy-proof mechanisms for maximum cost and the social cost objectives in a ciecle. We consider the case where agents are located on a circle with a known prelocated facility. In this setting, the impossibility results in Section \ref{sec3} still hold. We present a deterministic strategy-proof mechanism achieving $2$-approximation under the maximum cost objective. 

Suppose a facility has been prelocated at point 0 on a circle of length 2, as shown in Figure \ref{fig4}. Let $N = \{1, 2, \ldots, n\}$ be the set of agents, where each agent $i$ is located at a point $x_i \in [0, 2)$ on the circle, measured in the clockwise direction from 0. That is, $x_i$ denotes the clockwise arc length from the prelocated facility to agent $i$’s position. Define
\(
x_a = \max_{x_\in[0,1)}x_{i}, \quad
x_b = \min_{x_{i}\in [1,2)}x_i,
\)
where $x_a$ and $x_b$ are the locations of agents farthest from the prelocated facility.
\begin{figure}
    \centering
    \includegraphics[width=0.8\linewidth]{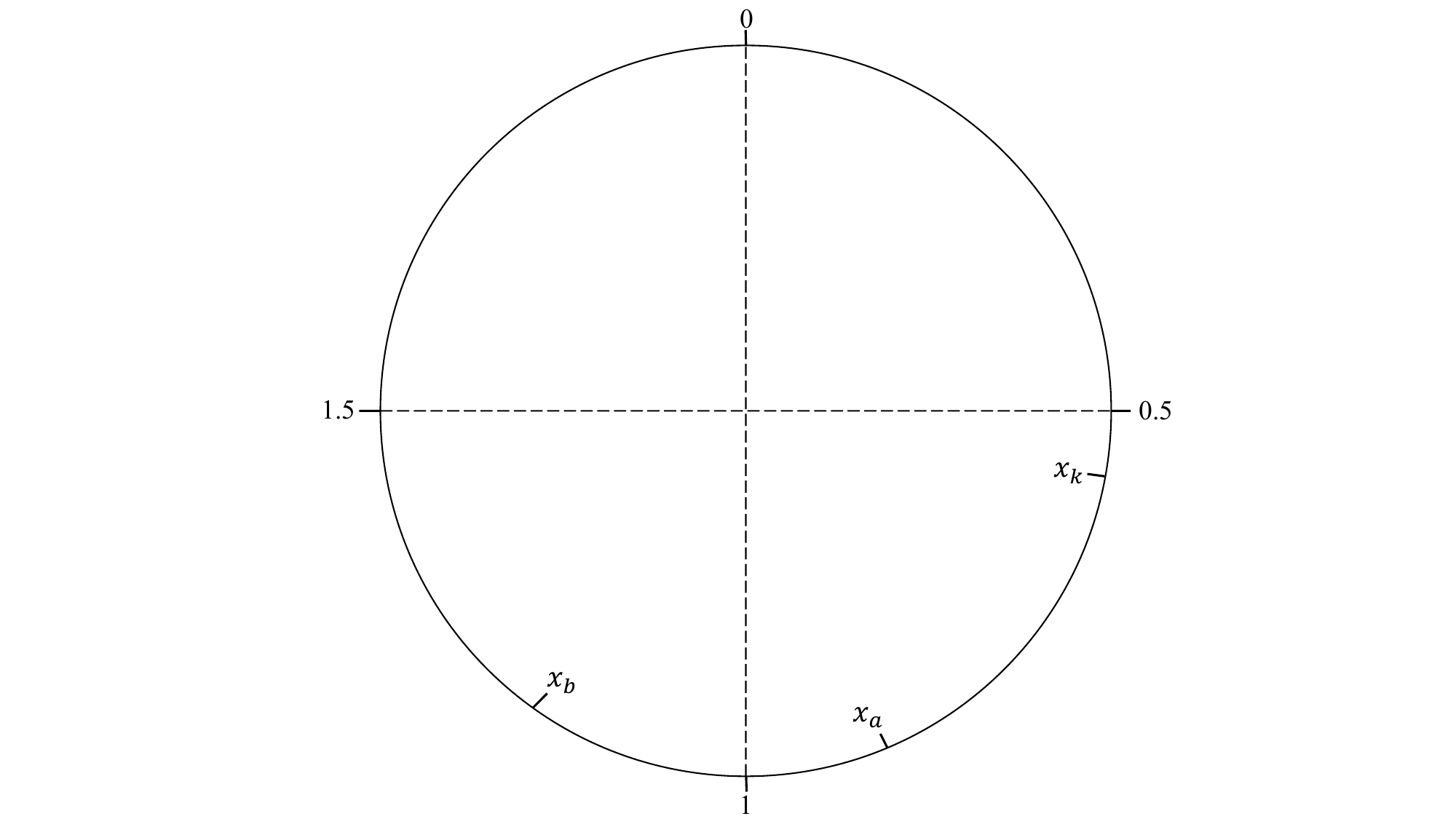}
    \caption{Illustration of the setting on the circle.}
    \label{fig4}
\end{figure}

\begin{mechanism}\label{mec4}
    Given $\mathbf{ x}=( x_{1}, x_{2},..., x_{n})$, if $x_{a}\geq 2-x_{b}$, the facility is located at $y$ as follows:
    $$y=\begin{cases}
    x_a &,\text{if } x_a\geq  \min\{4-2x_b,1\},\\
 \max\{2x_b-2,1\}&,\text{if }x_a\leq  \min\{4-2x_b,1\}.
    \end{cases}
    $$
If $x_{a}< 2-x_{b}$, $y$ is defined symmetrically.
\end{mechanism}
Following the proof idea of Theorem \ref{thm3}, it is easy to see that Mechanism \ref{mec4} is strategy-proof.
\begin{theorem}
    Mechanism \ref{mec4} is strategy-proof.
\end{theorem}
\begin{theorem}
    Mechanism \ref{mec4} is $2$-approximation under the maximum cost objective.
\end{theorem}

\begin{proof}
Denote Mechanism \ref{mec4} by $f$. For any instance $\mathbf{x}=(x_{1},x_{2},...,x_{n})$, without loss of generality, let $x_{a}\geq 2-x_{b}.$ If $x_{b}\geq1.5$, the proof follows the same idea as that of Theorem \ref{thm4}. Now consider $x_{b}\in [1,1.5)$, we have $MC(\mathbf{x},OPT)\leq MC(\mathbf{x},1) \leq\frac{1}{2}.$ If $|1-OPT(\mathbf{x})|\geq0.5$, it can be proved that $MC(\textbf{x},OPT)\geq\max\{cost(x_{a},OPT),cost(x_{b},OPT)\}>0.5$ which leads to a contradiction. Then we have $|1-OPT(\mathbf{x})|<0.5.$

Let $x_k$ be the location closest to either 0.5 or 1.5, whichever is closer. Without loss of generality, assume $x_{k}\in [0,1)$.

\textbf{Case 1} $x_{k}\geq 0.5$.
\begin{align*}
    \frac{MC(\mathbf{x},f)}{MC(\mathbf{x},OPT)} &= \frac{1-x_{k}}{MC(\mathbf{x},OPT)} \leq \frac{1-x_{k}}{\max\limits_{i=b,k}cost(x_{i},OPT)}\leq\frac{1-x_{k}}{\frac{x_b-1+1-x_{k}}{2}}<2.
\end{align*}

    \textbf{Case 2} $x_{k}< 0.5$.
\begin{align*}
    \frac{MC(\mathbf{x},f)}{MC(\mathbf{x},OPT)} &= \frac{x_{k}}{MC(\mathbf{x},OPT)} \leq \frac{x_{k}}{\max\limits_{i=b,k}cost(x_{i},OPT)}\leq\frac{x_{k}}{\frac{x_b-1+1-x_{k}}{2}}<2.
\end{align*}
The proof for the case when $x_k\in[1,2)$ is similar.
\end{proof}

\begin{theorem}
    Mechanism \ref{mec4} is $n$-approximation under the social cost objective.
\end{theorem}
\begin{proof}
    Denote Mechanism \ref{mec4} by $f$. For any instance $\mathbf{x}=(x_1,x_2,...,x_n)$, without loss of generality, let $x_a\geq 2-x_b$. If $x_a \leq0.5$, the proof follows the same idea as that of Theorem \ref{thm4}. If $x_a>0.5$, then we have $|1-f(\textbf{x})|<0.5.$ It is easy to know $cost(x_i,f(\textbf{x}))\leq\frac{3}{4}$. Then
    \begin{align*}
        &\frac{SC(\textbf{x},f)}{SC(\textbf{x},OPT)}\\&\leq 1+\frac{SC(\textbf{x},f)-SC(\textbf{x},OPT)}{SC(\textbf{x},OPT)}\\&\leq1+SC(\textbf{x},f)\\&\leq 1+\frac{3n}{4}\leq n.
    \end{align*}
\end{proof}
\begin{theorem}\rm{\cite{bib3}}
Proportional Mechanism is $6$-approximation under the social cost objective.
\end{theorem}

\section{Improvement Ratio}
In this section, we study a performance measure for mechanisms in facility location problems with prelocated facilities, called the improvement ratio which is defined
as the worst ratio between the improvement achieved by a mechanism and the best possible improvement after adding a new facility. 

Formally, we focus on the maximum cost objective, where the goal is to minimize the largest distance incurred by any agent to its nearest facility. Let $MC(\mathbf{x}, 0)$ denote the maximum cost when no additional facility is placed, and Let $OPT_{MC}$ denote the optimal placement of a new facility minimizing the maximum cost, and let $MC(\mathbf{x}, OPT_{MC})$ be the corresponding optimal value. A strategy-proof mechanism $f$ has an \textbf{improvement ratio} of $\alpha\textbf{ }(\geq1)$ under the maximum cost objective, if
$$\alpha=\sup_{\mathbf{x}\in\mathbb{R}^{n}}\frac{MC(\mathbf{x},0)-MC(\mathbf{x},OPT_{MC})}{MC(\mathbf{x},0)-MC(\mathbf{x},f)}.$$
The improvement ratio is defined similarly under the social cost objective. A mechanism with a lower improvement ratio indicates that it makes near-optimal use of the new facility.

We analyze the improvement ratio under the objective of maximum cost and the social cost in general and special settings in the real line.
\begin{theorem}\label{tthm1}
    Any deterministic strategy-proof mechanism has an improvement ratio of at least $\frac{3}{2}$ for the maximum cost.
\end{theorem}
\begin{proof}
 Assume there exists a deterministic strategy-proof mechanism $f$ with approximation ratio less than $\frac{3}{2}$. Consider an instance $\mathbf{x}=(x_{1},x_{2})=(1,2)$. It holds that $L(\mathbf{x})-MC(\mathbf{x},OPT)\leq \frac{3}{2}\cdot\left[L(\mathbf{x})-MC(\mathbf{x},f)\right]$, then $f(\mathbf{x})\in(1,2)$. Let $f(\mathbf{x})=1+\epsilon$ where $\epsilon\in(0,1)$.

  If $\epsilon\in\left(\frac{1}{3},1\right)$. Consider the instance $\mathbf{x'}=(x_{1},x_{2}')=(1,1+\epsilon)$, note that $OPT(\mathbf{x'})=1+\frac{\epsilon}{2}$ and $MC(\mathbf{x'},OPT)=\frac{\epsilon}{2}$. By strategy-proofness, we have $cost(x_{2}',f(\mathbf{x'}))\leq cost(x_{2}',f(\mathbf{x}))=0$. Therefore, $f(\mathbf{x'})=1+\epsilon$. Now consider the instance $\mathbf{x''}=(x_{1}',x_{2}')=(\frac{1+\epsilon}{2},1+\epsilon)$. It holds that $L(\mathbf{x''})-MC(\mathbf{x''},OPT)\leq \frac{3}{2}\cdot\left[L(\mathbf{x''})-MC(\mathbf{x''},f)\right]$, then $f(\mathbf{x''})\in \left(\frac{1+\epsilon}{2},1+\epsilon\right)$. We have $cost(x_{1},f(\mathbf{x''}))<\max\left\{\frac{1-\epsilon}{2},\epsilon\right\}=\epsilon, cost(x_{1},f(\mathbf{x'}))=1+\epsilon-1=\epsilon$. Therefore, agent $1$ can benefit by reporting $x_{1}$ to $x_{1}'$, which contradicts the strategy-proofness of $f$.

  If $\epsilon\in\big(0,\frac{1}{3}\big]$. Consider the instance $\mathbf{x'}=(x_{1}',x_{2})=(1+\epsilon,2)$. By strategy-proofness, we have $cost(x_{1}',f(\mathbf{x'}))\leq cost(x_{1}',f(\mathbf{x}))=0$. Therefore, $f(\mathbf{x'})=1+\epsilon$. Now consider the instance $\mathbf{x''}=(x_{1}',x_{2}')=(1+\epsilon,2+2\epsilon)$. It holds that $f(\mathbf{x''})\in \left(1+\epsilon,2+2\epsilon\right)$. We have $cost(x_{2},f(\mathbf{x}''))<\max\left\{2\epsilon,1-\epsilon\right\}=1-\epsilon, cost(x_{2},f(\mathbf{x'}))=2-(1+\epsilon)=1-\epsilon$. Therefore, agent $2$ can benefit by reporting $x_{2}$ to $x_{2}'$, which contradicts the strategy-proofness of $f$.
\end{proof}
\begin{theorem}
Mechanism \ref{mec1} is $1.5$-improvement under the maximum cost objective.
\end{theorem}

\begin{proof}
 Denote Mechanism \ref{mec1} by $f$. Without loss of generality, consider any instance $\mathbf{x}=(x_{1},x_{2},...,x_{3})$ with $L(x)=x_{n}\geq|x_{1}|$.

\textbf{Case 1}. $0\leq x_{1}\leq x_{n}$.
$$MC(\mathbf{x},OPT)=\max\left\{b(\mathbf{x}),\frac{L(\mathbf{x})-l(\mathbf{x})}{2}\right\}.$$

\textbf{Case 1.1}. $l(\mathbf{x})\geq\frac{L(\mathbf{x})}{2}$.

If $b(\mathbf{x})\geq L(\mathbf{x})-l(\mathbf{x})$,
\begin{align*}
&\frac{MC(\mathbf{x},0)-MC(\mathbf{x},OPT)}{MC(\mathbf{x},0)-MC(\mathbf{x},f)}\\&=
\frac{L(\mathbf{x})-MC(\mathbf{x},OPT)}{L(\mathbf{x})-MC(\mathbf{x},f)}\\&=\frac{L(\mathbf{x})-\max\left\{b(\mathbf{x}),\frac{L(\mathbf{x})-l(\mathbf{x})}{2}\right\}}{L(\mathbf{x})-\max\{b(\mathbf{x}),L(\mathbf{x})-l(\mathbf{x})\}}\\&=\frac{L(\mathbf{x})
-b(\mathbf{x})}{L(\mathbf{x})-b(\mathbf{x})}=1.
\end{align*}

If $b(\mathbf{x})< L(\mathbf{x})-l(\mathbf{x})$,
\begin{align*}
    \frac{MC(\mathbf{x},0)-MC(\mathbf{x},OPT)}{MC(\mathbf{x},0)-MC(\mathbf{x},f)}&=\frac{L(\mathbf{x})-MC(\mathbf{x},OPT)}{L(\mathbf{x})-MC(\mathbf{x},f)}\\
    &=\frac{L(\mathbf{x})-\max\left\{b(\mathbf{x}),\frac{L(\mathbf{x})-l(\mathbf{x})}{2}\right\}}{L(\mathbf{x})-\max\{b(\mathbf{x}),L(\mathbf{x})-l(\mathbf{x})\}} \\
    &\leq \frac{L(\mathbf{x})-\frac{L(\mathbf{x})-l(\mathbf{x})}{2}}{l(\mathbf{x})}\\
    &\leq\frac{L(\mathbf{x})+l(\mathbf{x})}{2l(\mathbf{x})} \\&
    \leq \frac{3}{2}.
\end{align*}

\textbf{Case 1.2}. $l(\mathbf{x})<\frac{L(\mathbf{x})}{2}$. Let $d(\mathbf{x})$ be the nearest position of agents to $\frac{L(\mathbf{x})}{2}$.
\begin{align*}
    \frac{MC(\mathbf{x},0)-MC(\mathbf{x},OPT)}{MC(\mathbf{x},0)-MC(\mathbf{x},f)}&=\frac{L(\mathbf{x})-MC(\mathbf{x},OPT)}{L(\mathbf{x})-MC(\mathbf{x},f)}\\&=\frac{L(\mathbf{x})-\max\left\{b(\mathbf{x}),\frac{L(\mathbf{x})-l(\mathbf{x})}{2}\right\}}{L(\mathbf{x})-\max\{b(\mathbf{x}),d(\mathbf{x})\}}\\
    &\leq\frac{L(\mathbf{x})-\frac{L(\mathbf{x})-l(\mathbf{x})}{2}}{L(\mathbf{x})-d(\mathbf{x})}\\
    &\leq\frac{L(\mathbf{x})+l(\mathbf{x})}{2L(\mathbf{x})-2l(\mathbf{x})}\\
    &\leq\frac{3}{2}.
\end{align*}

\textbf{Case 2}. $x_{1}<0\leq x_{n}$ and $2|x_{1}|\leq x_{n}$.
$$MC(\mathbf{x},OPT)=\max\left\{|x_{1}|,b(\mathbf{x}),\frac{L(\mathbf{x})-l(\mathbf{x})}{2}\right\}.$$

\textbf{Case 2.1}. $l(\mathbf{x})\geq\frac{L(\mathbf{x})}{2}$.

If $b(\mathbf{x})\geq L(\mathbf{x})-l(\mathbf{x})$,
\begin{align*}
    \frac{MC(\mathbf{x},0)-MC(\mathbf{x},OPT)}{MC(\mathbf{x},0)-MC(\mathbf{x},f)}&=\frac{L(\mathbf{x})-MC(\mathbf{x},OPT)}{L(\mathbf{x})-MC(\mathbf{x},f)}\\&=\frac{L(\mathbf{x})-\max\left\{b(\mathbf{x}),\frac{L(\mathbf{x})-l(\mathbf{x})}{2}\right\}}{L(\mathbf{x})-\max\{b(\mathbf{x}),L(\mathbf{x})-l(\mathbf{x})\}}\\&=\frac{L(\mathbf{x})-b(\mathbf{x})}{L(\mathbf{x})-b(\mathbf{x})}=1.
\end{align*}

If $b(\mathbf{x})< L(\mathbf{x})-l(\mathbf{x})$,
\begin{align*}
    \frac{MC(\mathbf{x},0)-MC(\mathbf{x},OPT)}{MC(\mathbf{x},0)-MC(\mathbf{x},f)}&=\frac{L(\mathbf{x})-MC(\mathbf{x},OPT)}{L(\mathbf{x})
-MC(\mathbf{x},f)}\\&=\frac{L(\mathbf{x})-\max\left\{b(\mathbf{x}),\frac{L(\mathbf{x})-l(\mathbf{x})}{2}\right\}}{L(\mathbf{x})-\max\{b(\mathbf{x}),L(\mathbf{x})-l(\mathbf{x})\}} \\
&\leq \frac{L(\mathbf{x})-\frac{L(\mathbf{x})-l(\mathbf{x})}{2}}{l(\mathbf{x})}\\
&\leq\frac{L(\mathbf{x})+l(\mathbf{x})}{2l(\mathbf{x})}
\\&\leq \frac{3}{2}.
\end{align*}

\textbf{Case 2.2}. $l(\mathbf{x})<\frac{L(\mathbf{x})}{2}$. Let $d(\mathbf{x})$ be the nearest position of agents to $\frac{L(\mathbf{x})}{2}$.
\begin{align*}
    \frac{MC(\mathbf{x},0)-MC(\mathbf{x},OPT)}{MC(\mathbf{x},0)-MC(\mathbf{x},f)}&=\frac{L(\mathbf{x})-MC(\mathbf{x},OPT)}{L(\mathbf{x})-MC(\mathbf{x},f)}\\&=\frac{L(\mathbf{x})-\max\left\{b(\mathbf{x}),\frac{L(\mathbf{x})-l(\mathbf{x})}{2}\right\}}{L(\mathbf{x})-\max\{b(\mathbf{x}),d(\mathbf{x})\}}\\
    &\leq \frac{L(\mathbf{x})-\frac{L(\mathbf{x})-l(\mathbf{x})}{2}}{L(\mathbf{x})-d(\mathbf{x})}\\
    &\leq\frac{L(\mathbf{x})+l(\mathbf{x})}{2L(\mathbf{x})-2l(\mathbf{x})}\\
    &\leq\frac{3}{2}.
\end{align*}

\textbf{Case 3}. $x_{1}<0\leq x_{n}$ and $2|x_{1}|> x_{n}$, then $f(\mathbf{x})=2x_{1}$.
\begin{align*}
    \frac{MC(\mathbf{x},0)-MC(\mathbf{x},OPT)}{MC(\mathbf{x},0)-MC(\mathbf{x},f)}&=\frac{L(\mathbf{x})-MC(\mathbf{x},OPT)}{L(\mathbf{x})-MC(\mathbf{x},f)}\leq\frac{L(\mathbf{x})-|x_{1}|}{L(\mathbf{x})-|x_{1}|}=1.
\end{align*}
The proof is complete.
\end{proof}

\begin{theorem}\label{tthm2}
    Any randomized strategy-proof mechanism has an improvement ratio of at least $1.01$ for the maximum cost.
\end{theorem}
\begin{proof}
   Let $f$ be any randomized strategy-proof mechanism. Let $M>0$ be sufficiently large. Consider an instance $\mathbf{x}=(x_{1},x_{2})=(7+\sqrt{39},8+\sqrt{39})$. Obviously, $\sum_{i\in\{1,2\}} cost(x_{i},f(\mathbf{x}))=E_{Y\sim f( \mathbf{x})}\left[cost(x_{1},Y)+cost(x_{2},Y)\right]\geq1$. Without loss of generality, assume that  $cost(x_{1},f( \mathbf{x}))\geq\frac{1}{2}$.

   Now consider another instance $ \mathbf{x'}=(x_{1}',x_{2}')=(6+\sqrt{39},8+\sqrt{39})$. Note that $OPT( \mathbf{x'})=7+\sqrt{39}$ and $MC( \mathbf{x'},OPT)=1$. By strategy-proofness, we have $cost(7+\sqrt{39},f( \mathbf{x'}))=cost(x_{1},f(\mathbf{x'}))\geq cost(x_{1},f( \mathbf{x}))\geq \frac{1}{2}$. Otherwise, the agent located at $x_{1}$ in $\mathbf{x}$ can benefit by misreporting her location as $x_{1}'=6+\sqrt{39}$. The maxmum cost of $f$ w.r.t. $\mathbf{x'}$ is
\begin{align*}
    MC( \mathbf{x'},f) =& E_{Y^{'}\sim f( \mathbf{x}')}\left[\max_{i\in\{1,2\}}cost(x_{i}',Y')\right],\\
    =&Pr\{Y'\leq12+2\sqrt{39}\}\cdot E_{Y^{'}\sim f( \mathbf{x}')}\left[\max_{i\in\{1,2\}}cost(x_{i}',Y')|Y'\leq12+2\sqrt{39}\right]\\
    &+Pr\{Y'>12+2\sqrt{39}\}\cdot E_{Y^{'}\sim f( \mathbf{x}')}\left[\max_{i\in\{1,2\}}cost(x_{i}',Y')|Y'>12+2\sqrt{39}\right].
\end{align*}

Note that $\max_{i\in\{1,2\}}cost(x_{i}',Y')=1+cost(M+1,Y')$ when $Y'\leq12+2\sqrt{39}$ and $\max_{i\in\{1,2\}}cost(x_{i}',Y')>6+\sqrt{39}\geq cost(7+\sqrt{39},Y')-1$ when $Y'>12+2\sqrt{39}$. We analyze $MC(\mathbf{x'},f)$ according to the following cases.

\textbf{Case 1} $Pr\{Y'>12+2\sqrt{39}\}\geq\frac{3}{12+2\sqrt{39}}$.
\begin{align*}
MC( \mathbf{x'},f)&\geq Pr\{Y'>12+2\sqrt{39}\}\cdot E_{Y^{'}\sim f( \mathbf{x}')}\left[\max_{i\in\{1,2\}}cost(x_{i}',Y')|Y'>12+2\sqrt{39}\right]\\
&\geq\frac{3}{12+2\sqrt{39}}\cdot (6+\sqrt{39})=\frac{3}{2}.
\end{align*}

\textbf{Case 2} $Pr\{Y'<12+2\sqrt{39}\}\geq\frac{3}{12+2\sqrt{39}}$.
\begin{align*}
    MC( \mathbf{x'},f) \geq&Pr\{Y'\leq12+2\sqrt{39}\}\cdot E_{Y^{'}\sim f( \mathbf{x}')}\left[1+cost(7+\sqrt{39},Y')|Y'\leq12+2\sqrt{39}\right] \\
    &+Pr\{Y'>12+2\sqrt{39}\}\cdot E_{Y^{'}\sim f( \mathbf{x}')}\left[1+cost(7+\sqrt{39},Y')-2|Y'>12+2\sqrt{39}\right]\\
    =& 1+cost(7+\sqrt{39},Y')-2Pr\{Y'>12+2\sqrt{39}\}\\
    \geq& \frac{3}{2}-\frac{3}{6+\sqrt{39}}.
\end{align*}

Therefore,
$$\frac{MC(\mathbf{x},0)-MC(\mathbf{x},OPT)}{\max_{i\in N}{|x_i|}-MC(\mathbf{x},f)}\geq\frac{26+4\sqrt{39}}{25+4\sqrt{39}+\frac{6}{6+\sqrt{39}}}=\frac{598+104\sqrt{39}}{1235}>1.01.$$
\end{proof}
\begin{theorem}\label{tthm3}
Any deterministic strategy-proof mechanism has an improvement ratio of at least $1.5$ for the social cost.
\end{theorem}
\begin{proof}
  Assume there exists a deterministic strategy-proof mechanism $f$ with approximation ratio less than 1.5. Consider an instance $\mathbf{x}=(x_{1},x_{2})=(-1,1).$ Without loss of generality, let $f(\mathbf{x})=y\leq0$, we have $cost(x_{2},y)=1$. Consider the instance $\mathbf{x}'=(x_{1},x_{2}')=(-1,1.5)$ and $f(\mathbf{x}')=y'$. Note that $SC(\mathbf{x'},OPT)=1$, we have $\sum_{i\in\{1,2\}}|x_{i}|-SC(\mathbf{x'},OPT)<1.5\cdot\left[sum_{i\in\{1,2\}}|x_{i}|-SC(\mathbf{x}',y')\right]$, then $SC(\mathbf{x}',f)<1.5$. It can lead to the $cost(x_{2},y')<1$. Agent 2 can benefit by reporting $ x_{2}$ to $x_{2}'$, which contradicts with the strategy-proofness of $f$.
\end{proof}
\begin{lemma}\label{llem1}
  For any instance $\mathbf{x}=(x_{1},...,x_{n})$ with $x_{1}\leq x_{2}\leq ...\leq x_{n}$, define
  $$N_{1}=\left\{ i\big| |x_{i}-OPT_{SC}(\mathbf{x})|\leq |x_{i}|\text{ }and\text{ }x_{i}\leq OPT_{SC}(\mathbf{x}) \right\},$$
  $$N_{2}=\left\{ i\big| |x_{i}-OPT_{SC}(\mathbf{x})|\leq |x_{i}|\text{ }and\text{ }x_{i}> OPT_{SC}(\mathbf{x}) \right\}.$$
  We have $|N_{1}|=|N_{2}|.$
\end{lemma}
\begin{proof}
   For any instance $\mathbf{x}=(x_{1},...,x_{n})$, if $|N_1|>|N_2|$, then it is easy to prove that there exists a location $p$ such that 
   $$SC(\textbf{x},p)<SC(\textbf{x},OPT_{SC}),$$
   which is contradicts the optimality. The case of $|N_1|<|N_2|$ can be proved similarly.
\end{proof}
By Lemma \ref{llem1}, the following Lemma \ref{llem2} naturally follows:
\begin{lemma}\label{llem2}
      Define $\mathbf{\widetilde{x}}=\left\{ x_{i} \large{|} |x_{i}-OPT_{SC}(\mathbf{x})|\leq |x_{i}|\right\}=(x_{k_{1}},x_{k_{2}},...,x_{k_{m}}),x_{k_{1}}\leq x_{k_{2}}\leq...\leq x_{k_{m}}.$ We have $$OPT_{SC}(\mathbf{x})=med(\mathbf{\widetilde{x}})=x_{\lceil\frac{k_{m}}{2}\rceil}.$$
\end{lemma}
\begin{theorem}\label{tthm5}
Mechanism \ref{mec1} is a deterministic strategy-proof $\frac{n}{2}$-improvement under the social cost objective.
\end{theorem}
\begin{proof}
  Denote Mechanism 1 by $f$. Without loss of generality, consider any instance $\mathbf{x}=(x_{1},x_{2},...,x_{3})$ with $L(x)=x_{n}\geq|x_{1}|$,

\textbf{Case 1}. $x_{1}\geq\frac{L(\mathbf{x})}{2}$. By Lemma \ref{llem2}, we have $\mathbf{x}=\mathbf{\widetilde{x}}$,
\begin{align*}
SC( \mathbf{x},OPT) &= \sum_{i\leq\lceil\frac{n}{2}\rceil}[OPT(x)-x_{i}]+\sum_{i>\lceil\frac{n}{2}\rceil}[x_{i}-OPT(\mathbf{x})],\\
SC( \mathbf{x},f) &= \sum_{i\in N}(L(\mathbf{x})-x_{i}),
\end{align*}
\begin{align*}
&\frac{SC(\mathbf{x},0)-SC(\mathbf{x},OPT)}{SC(\mathbf{x},0)-SC(\mathbf{x},f)}\\
&=\frac{\sum_{i\in N}x_{i}-SC(\mathbf{x},OPT)}{\sum_{i\in N}x_{i}-SC(\mathbf{x},f)} \\&= \frac{\sum_{i\in N}x_{i}-\left\{\sum_{i\leq\lceil\frac{n}{2}\rceil}[OPT(x)-x_{i}]+\sum_{i>\lceil\frac{n}{2}\rceil}[x_{i}-OPT(\mathbf{x})]\right\}}{\sum_{i\in N}x_{i}-\left[\sum_{i\in N}(L(\mathbf{x})-x_{i})\right]}\\
&= \frac{\sum_{i\leq\lceil\frac{n}{2}\rceil}(2x_{i}-OPT(\mathbf{x}))+\sum_{i>\lceil\frac{n}{2}\rceil}OPT(\mathbf{x})}{\sum_{i\in N}(2x_{i}-L(\mathbf{x}))}\\
&= \frac{\sum_{i<\lceil\frac{n}{2}\rceil}2x_{i}}{\sum_{i<\lceil\frac{n}{2}\rceil}2x_{i}+\sum_{i\geq\lceil\frac{n}{2}\rceil}2x_{i}-n\cdot L(\mathbf{x})}\\
&= 1+\frac{n\cdot L(\mathbf{x})-\sum_{i\geq\lceil\frac{n}{2}\rceil}2x_{i}}{\sum_{i<\lceil\frac{n}{2}\rceil}2x_{i}+\sum_{i\geq\lceil\frac{n}{2}\rceil}2x_{i}-n\cdot L(\mathbf{x})}\\
&\leq 1+\frac{n\cdot L(\mathbf{x})-2L(\mathbf{x})-\frac{n-2}{2}\cdot L(\mathbf{x})}{\frac{n}{2}\cdot L(\mathbf{x})+\frac{n-2}{2}\cdot L(\mathbf{x})+2L(\mathbf{x})-n\cdot L(\mathbf{x})}\\&=\frac{n}{2}.
\end{align*}

\textbf{Case 2}. $x_{1}<\frac{L(\mathbf{x})}{2}$. If $x_{1}\geq\frac{OPT(\mathbf{x})}{2}$, it is evident that $\mathbf{x}=\mathbf{\widetilde{x}}$, the proof is similar to that of Case 1. If $x_{1}<\frac{OPT(\mathbf{x})}{2}$, define
\begin{align*}
  N_{1} &= \left\{i\big|x_{1}\leq x_{i}\leq\frac{OPT(\mathbf{x})}{2}\right\}, \\
  N_{2} &=
  \begin{cases}
  \left\{i\big|\frac{OPT(\mathbf{x})}{2}< x_{i}\leq\frac{L(\mathbf{x})}{2}\right\}, & \mbox{if } OPT(\mathbf{x})\geq\frac{L(\mathbf{x})}{2} \\
  \left\{i\big|\frac{OPT(\mathbf{x})}{2}< x_{i}\leq OPT(\mathbf{x}),i\leq \lceil\frac{k_{m}}{2}\rceil\right\}, & \mbox{if } OPT(\mathbf{x})<\frac{L(\mathbf{x})}{2}.
  \end{cases}
\\
  N_{3} &=
  \begin{cases}
  \left\{i\big|\frac{L(\mathbf{x})}{2}< x_{i}\leq OPT(\mathbf{x}),i\leq \lceil\frac{k_{m}}{2}\rceil\right\}, & \mbox{if } OPT(\mathbf{x})\geq\frac{L(\mathbf{x})}{2} \\
  \left\{i\big|OPT(\mathbf{x})\geq x_{i}\leq\frac{L(\mathbf{x})}{2},i>\lceil\frac{k_{m}}{2}\rceil\right\}, & \mbox{if } OPT(\mathbf{x})<\frac{L(\mathbf{x})}{2}.
  \end{cases}
\\
  N_{4} &= N-N_{1}-N_{2}-N_{3}
\end{align*}

We have
$$SC(\mathbf{x},OPT)=\sum_{i\in N_{1}}x_{i}+\sum_{i\notin N_{1}}|x_{i}-OPT(\mathbf{x})|.$$

\textbf{Case 2.1}. $OPT(\mathbf{x})\geq\frac{L(\mathbf{x})}{2}$. By Lemma \ref{llem1}, we have $|N_{2}|+|N_{3}|=|N_{4}|$,
$$SC(\mathbf{x},f) = \sum_{i\in N_{1}\cup N_{2}}x_{i}+\sum_{i\in N_{3}\cup N_{4}}(L(\mathbf{x})-x_{i}).$$
\begin{align*}
&\frac{SC(\mathbf{x},0)-SC(\mathbf{x},OPT)}{SC(\mathbf{x},0)-SC(\mathbf{x},f)}\\&=\frac{\sum_{i\in N}x_{i}-SC(\mathbf{x},OPT)}{\sum_{i\in N}x_{i}-SC(\mathbf{x},f)} \\&= \frac{\sum_{i\in N}x_{i}-\sum_{i\in N_{1}}x_{i}-\sum_{i\notin N_{1}}|x_{i}-OPT(\mathbf{x})|}{\sum_{i\in N}x_{i}-\sum_{i\in N_{1}\cup N_{2}}x_{i}-\sum_{i\in N_{3}\cup N_{4}}(L(\mathbf{x})-x_{i})}\\
&= \frac{\sum_{i\in N}x_{i}-\sum_{i\in N_{1}}x_{i}-\sum_{i\in N_{2}\cup N_{3}}(OPT(\mathbf{x})-x_{i})-\sum_{i\in N_{4}}(x_{i}-OPT(\mathbf{x}))}{\sum_{i\in N_{3}\cup N_{4}}\left(2x_{i}-L(\mathbf{x})\right)}\\
&= \frac{\sum_{i\in N_{2}\cup N_{3}}(2x_{i}-OPT(\mathbf{x}))+\sum_{i\in N_{4}}OPT(\mathbf{x})}{\sum_{i\in N_{3}\cup N_{4}}\left(2x_{i}-L(\mathbf{x})\right)}    \\
&\leq \frac{\sum_{i\in N_{3}}2x_{i}+|N_{2}|\cdot L(\mathbf{x})}{\sum_{i\in N_{3}}2x_{i}+\sum_{i\in N_{4}}2x_{i}-|N_{3}\cup N_{4}|\cdot L(\mathbf{x})}\\
&\leq \frac{\sum_{i\in N_{3}}2x_{i}+|N_{2}|\cdot L(\mathbf{x})}{\sum_{i\in N_{3}}2x_{i}+|N_{2}|\cdot L(\mathbf{x})+\sum_{i\in N_{4}}2x_{i}-2|N_{4}|\cdot L(\mathbf{x})}\\
&\leq \frac{\sum_{i\in N_{3}}2x_{i}+|N_{2}|\cdot L(\mathbf{x})}{\sum_{i\in N_{3}}2x_{i}+|N_{2}|\cdot L(\mathbf{x})+(|N_{4}|-1)\cdot L(\mathbf{x})+2L(\mathbf{x})-2|N_{4}|\cdot L(\mathbf{x})}\\
&\leq 1+\frac{2|N_{4}|\cdot L(\mathbf{x})-(|N_{4}|-1)\cdot L(\mathbf{x})-2L(\mathbf{x})}{\sum_{i\in N_{3}}2x_{i}+|N_{2}|\cdot L(\mathbf{x})+(|N_{4}|-1)\cdot L(\mathbf{x})+2L(\mathbf{x})-2|N_{4}|\cdot L(\mathbf{x})}\\
&\leq 1+\frac{2|N_{4}|-(|N_{4}|-1)-2}{|N_{3}|+|N_{2}|+|N_{4}|-1+2-2|N_{4}|}\\&=|N_{4}|\leq \frac{n-1}{2}\leq\frac{n}{2}.
\end{align*}
\textbf{Case 2.2}. $OPT(\mathbf{x})<\frac{L(\mathbf{x})}{2}$. By Lemma \ref{llem1}, we have $|N_{2}|=|N_{3}|+|N_{4}|$,
\begin{align*}
SC(\mathbf{x},f) &= \sum_{i\in N_{1}\cup N_{2}\cup N_{3}}x_{i}+\sum_{i\in N_{4}}(L(\mathbf{x})-x_{i}),\\
\frac{SC(\mathbf{x},0)-SC(\mathbf{x},OPT)}{SC(\mathbf{x},0)-SC(\mathbf{x},f)} &= \frac{\sum_{i\in N}x_{i}-\sum_{i\in N_{1}}x_{i}-\sum_{i\notin N_{1}}|x_{i}-OPT(\mathbf{x})|}{\sum_{i\in N}x_{i}-\sum_{i\in N_{1}\cup N_{2}\cup N_{3}}x_{i}-\sum_{i\in N_{4}}(L(\mathbf{x})-x_{i})}\\
&= \frac{\sum_{i\in N_{2}}2x_{i}}{\sum_{i\in N_{4}}\left(2x_{i}-L(\mathbf{x})\right)}\leq \frac{|N_{2}|\cdot L(\mathbf{x})}{L(\mathbf{x})}\leq\frac{n-1}{2}\leq\frac{n}{2}.
\end{align*}
\end{proof}
\begin{theorem}\label{tthm4}
Mechanism \ref{mec3} is a randomized strategy-proof $\frac{15}{13}$-improvement under the maximum cost objective.
\end{theorem}
\begin{proof}
  Denote Mechanism \ref{mec2} by $f$. Without loss of generality, consider any instance $\mathbf{x}=(x_{1},x_{2},...,x_{3})$ with $L(\mathbf{x})=x_{n}\geq x_{1} \geq0$. We have
  $$MC(\mathbf{x},OPT)=\max\left\{b(\mathbf{x}),\frac{L(\mathbf{x})-l(\mathbf{x})}{2}\right\}.$$

\textbf{Case 1}. $b(\mathbf{x})\geq L(\mathbf{x})-l(\mathbf{x})$.
$$\frac{L(\mathbf{x})-MC(\mathbf{x},OPT)}{L(\mathbf{x})-MC(\mathbf{x},f)}\leq\frac{L(\mathbf{x})-b(\mathbf{x})}{L(\mathbf{x})-b(\mathbf{x})}=1<\frac{15}{13}.$$

\textbf{Case 2}. $b(\mathbf{x})<L(\mathbf{x})-l(\mathbf{x})$ and $l(\mathbf{x})<\frac{2L(\mathbf{x})}{3}$.

Let $l(\mathbf{x}) = \frac{L(\mathbf{x})}{3}+ \Delta \leq \frac{2L(\mathbf{x})}{3}, \Delta \in (0, \frac{L(\mathbf{x})}{3}]$. We have
  \begin{align*}
  MC(\mathbf{x},OPT)&=\max\left\{b(\mathbf{x}),\frac{L(\mathbf{x})-l(\mathbf{x})}{2}\right\}=\max\left\{b(\mathbf{x}),\frac{L(\mathbf{x})}{3}-\frac{\Delta}{2}\right\},\\
  MC( \mathbf{x},f) &\leq \frac{1}{6}\cdot \frac{L( \mathbf{x})}{3}+\frac{1}{3}\cdot\max\left\{b( \mathbf{x}),cost(l( \mathbf{x}),\frac{5L(\mathbf{x})}{6} )\right\}+\frac{1}{2}\cdot (L( \mathbf{x})-l( \mathbf{x}))\\
  &= \frac{L(\mathbf{x})}{18}+\frac{1}{3}\cdot\max\left\{b( \mathbf{x}),\frac{L( \mathbf{x})}{2}-\Delta\right\}+\frac{L(\mathbf{x})}{3}-\frac{\Delta}{2}.
  \end{align*}

\textbf{Case 2.1} $b( \mathbf{x}) \geq \frac{L(\mathbf{x})}{2} - \Delta \in \left[\frac{L(\mathbf{x})}{6}, \frac{L(\mathbf{x})}{2}\right)$.
\begin{align*}
&\frac{L(\mathbf{x})-MC(\mathbf{x},OPT)}{L(\mathbf{x})-MC(\mathbf{x},f)} \\&
\leq\frac{L(\mathbf{x})-\max\left\{b(\mathbf{x}),\frac{L(\mathbf{x})}{3}-\frac{\Delta}{2}\right\}}{L(\mathbf{x})-\left[\frac{L(\mathbf{x})}{18}+\frac{1}{3}b(\mathbf{x})+\frac{L(\mathbf{x})}{3}-\frac{\Delta}{2}\right]} \\&
=\frac{L(\mathbf{x})-\max\left\{b(\mathbf{x}),\frac{L(\mathbf{x})}{3}-\frac{\Delta}{2}\right\}}{\frac{11}{18}L(\mathbf{x})-\frac{b(\mathbf{x})}{3}+\frac{\Delta}{2}}\\
\leq&\frac{L(\mathbf{x})-b(\mathbf{x})}{\frac{11}{18}L(\mathbf{x})-\frac{b(\mathbf{x})}{3}+\frac{L(\mathbf{x})}{4}-\frac{b(\mathbf
x)}{2}}\\&
=\frac{L(\mathbf{x})-b(\mathbf{x})}{\frac{31}{36}L(\mathbf{x})-\frac{5}{6}b(\mathbf{x})}\\&
\leq\frac{1-\frac{1}{2}}{\frac{31}{35}-\frac{5}{6}\cdot\frac{1}{2}}<\frac{15}{13}.
\end{align*}

\textbf{Case 2.2} $b( \mathbf{x}) < \frac{L(\mathbf{x})}{2} - \Delta$.
\begin{align*}
&\frac{L(\mathbf{x})-MC(\mathbf{x},OPT)}{L(\mathbf{x})-MC(\mathbf{x},f)}\\&\leq\frac{L(\mathbf{x})-\max\left\{b(\mathbf{x}),\frac{L(\mathbf{x})}{3}-\frac{\Delta}{2}\right\}}{L(\mathbf{x})-\left\{\frac{L(\mathbf{x})}{18}+\frac{1}{3}\left[\frac{L(\mathbf{x})}{2}-\Delta\right]+\frac{L(\mathbf{x})}{3}-\frac{\Delta}{2}\right\}}\\&=\frac{L(\mathbf{x})-\max\left\{b(\mathbf{x}),\frac{L(\mathbf{x})}{3}-\frac{\Delta}{2}\right\}}{\frac{4}{9}L(\mathbf{x})+\frac{5}{6}\Delta}\\
&\leq\frac{L(\mathbf{x})-b(\mathbf{x})}{\frac{4}{9}L(\mathbf{x})+\frac{5}{6}\Delta}\\&\leq\frac{L(\mathbf{x})-\frac{L(\mathbf{x})}{2}+\Delta}{\frac{4}{9}L(\mathbf{x})+\frac{5}{6}\Delta}\\&=\frac{9L(\mathbf{x})+18\Delta}{8L(\mathbf{x})+15\Delta}\leq\frac{9}{8}<\frac{15}{13}.
\end{align*}

\textbf{Case 3}. $b(\mathbf{x})<L(\mathbf{x})-l(\mathbf{x})$ and $l(\mathbf{x})\geq\frac{2L(\mathbf{x})}{3}$.
\begin{align*}
MC( \mathbf{x},OPT) &= \max\left\{b( \mathbf{x}), \frac{L( \mathbf{x}) - l( \mathbf{x})}{2}\right\}, \\
MC( \mathbf{x},f) &= \frac{2}{3} \cdot \max\left\{b( \mathbf{x}), L( \mathbf{x}) - l( \mathbf{x})\right\}+\frac{1}{3}\max\left\{b(\mathbf{x}),\frac{L(\mathbf{x})-l(\mathbf{x})}{2}\right\},\\
\frac{L(\mathbf{x})-MC(\mathbf{x},OPT)}{L(\mathbf{x})-MC(\mathbf{x},f)} &= \frac{L(\mathbf{x})-\max\left\{b( \mathbf{x}), \frac{L( \mathbf{x}) - l( \mathbf{x})}{2}\right\}}{L(\mathbf{x})-\frac{2}{3}\left[L( \mathbf{x}) - l( \mathbf{x})\right]-\frac{1}{3}\max\left\{b(\mathbf{x}),\frac{L(\mathbf{x})-l(\mathbf{x})}{2}\right\}}\\
&=\frac{L(\mathbf{x})-\max\left\{b(\mathbf{x}),\frac{L(\mathbf{x})-l(\mathbf{x})}{2}\right\}+2l(\mathbf{x})-2l(\mathbf{x})}{\frac{L(\mathbf{x})}{3}-\frac{1}{3}\max\left\{b(\mathbf{x}),\frac{L(\mathbf{x})-l(\mathbf{x})}{2}\right\}+\frac{2l(\mathbf{x})}{3}}\\
&\leq 3-\frac{6l(\mathbf{x})}{\frac{L(\mathbf{x})}{2}+\frac{5l(\mathbf{x})}{2}}=3-\frac{12l(\mathbf{x})}{L(\mathbf{x})+5l(\mathbf{x})}\leq3-\frac{8}{1+\frac{10}{3}}=\frac{15}{13}.
\end{align*}
\end{proof}
\section{Conclusions and Open Problems}\label{sec5}


This paper studied strategy-proof mechanisms for locating a new homogeneous facility when a facility was already prelocated, with agents on a real line or a circle. Agents’ private locations determined their cost as the distance to the nearest facility. We considered two settings: the general setting where agents could be on both sides of the prelocated facility, and the special setting where all agents lay on the same side.

For the line in the general setting, we designed a deterministic mechanism achieving a tight 2-approximation for the maximum cost, and proved a lower bound of $1.5 - \epsilon$ for any randomized mechanism. For social cost, the deterministic mechanisms had an upper bound of $n$ and lower bounds of 1.5 (deterministic) and 1.0425 (randomized). In the special setting, we provided a randomized mechanism that achieved an approximation ratio $5/3$ for maximum cost, and a deterministic mechanism achieved $n-1$ for the social cost. On the circle, we provided a deterministic mechanism with 2-approximation for the maximum cost.

We also proposed the improvement ratio, a performance measure specially designed for mechanisms in facility location problems with prelocated facilities. For the maximum cost, we established tight bounds of 1.5 for deterministic mechanisms in the general setting, and an upper bound of $15/13$ for randomized mechanisms in the special setting. For social cost in the special setting, we gave an upper bound of $n/2$.

For future work, a natural direction was to narrow down the gap between the upper bound and the lower bound for strategy-proof mechanisms in our model. Our model could also be extended to other settings with various facility preferences, such as obnoxious facilities, or where agents had different preferences between the prelocated facility and the newly established one. It was also worthwhile to investigate locating new facilities under a certain constraint on facilities.

\bmhead{Acknowledgements}
 This research was supported in part by the National Natural Science Foundation of China (12201590, 12171444) and Natural Science Foundation of Shandong Province (ZR2024MA031).
 
\bmhead{Data availability}
None.

\section*{Declarations}
\bmhead{Conflict of interest}
The authors declare that they have no conflict of interest.



\backmatter


\begin{thebibliography}{22}
\ifx \bisbn   \undefined \def \bisbn  #1{ISBN #1}\fi
\ifx \binits  \undefined \def \binits#1{#1}\fi
\ifx \bauthor  \undefined \def \bauthor#1{#1}\fi
\ifx \batitle  \undefined \def \batitle#1{#1}\fi
\ifx \bjtitle  \undefined \def \bjtitle#1{#1}\fi
\ifx \bvolume  \undefined \def \bvolume#1{\textbf{#1}}\fi
\ifx \byear  \undefined \def \byear#1{#1}\fi
\ifx \bissue  \undefined \def \bissue#1{#1}\fi
\ifx \bfpage  \undefined \def \bfpage#1{#1}\fi
\ifx \blpage  \undefined \def \blpage #1{#1}\fi
\ifx \burl  \undefined \def \burl#1{\textsf{#1}}\fi
\ifx \doiurl  \undefined \def \doiurl#1{\url{https://doi.org/#1}}\fi
\ifx \betal  \undefined \def \betal{\textit{et al.}}\fi
\ifx \binstitute  \undefined \def \binstitute#1{#1}\fi
\ifx \binstitutionaled  \undefined \def \binstitutionaled#1{#1}\fi
\ifx \bctitle  \undefined \def \bctitle#1{#1}\fi
\ifx \beditor  \undefined \def \beditor#1{#1}\fi
\ifx \bpublisher  \undefined \def \bpublisher#1{#1}\fi
\ifx \bbtitle  \undefined \def \bbtitle#1{#1}\fi
\ifx \bedition  \undefined \def \bedition#1{#1}\fi
\ifx \bseriesno  \undefined \def \bseriesno#1{#1}\fi
\ifx \blocation  \undefined \def \blocation#1{#1}\fi
\ifx \bsertitle  \undefined \def \bsertitle#1{#1}\fi
\ifx \bsnm \undefined \def \bsnm#1{#1}\fi
\ifx \bsuffix \undefined \def \bsuffix#1{#1}\fi
\ifx \bparticle \undefined \def \bparticle#1{#1}\fi
\ifx \barticle \undefined \def \barticle#1{#1}\fi
\bibcommenthead
\ifx \bconfdate \undefined \def \bconfdate #1{#1}\fi
\ifx \botherref \undefined \def \botherref #1{#1}\fi
\ifx \url \undefined \def \url#1{\textsf{#1}}\fi
\ifx \bchapter \undefined \def \bchapter#1{#1}\fi
\ifx \bbook \undefined \def \bbook#1{#1}\fi
\ifx \bcomment \undefined \def \bcomment#1{#1}\fi
\ifx \oauthor \undefined \def \oauthor#1{#1}\fi
\ifx \citeauthoryear \undefined \def \citeauthoryear#1{#1}\fi
\ifx \endbibitem  \undefined \def \endbibitem {}\fi
\ifx \bconflocation  \undefined \def \bconflocation#1{#1}\fi
\ifx \arxivurl  \undefined \def \arxivurl#1{\textsf{#1}}\fi
\csname PreBibitemsHook\endcsname

\bibitem[\protect\citeauthoryear{Procaccia and Tennenholtz}{2013}]{bib1}
\begin{barticle}
\bauthor{\bsnm{Procaccia}, \binits{A.D.}},
\bauthor{\bsnm{Tennenholtz}, \binits{M.}}:
\batitle{Approximate mechanism design without money}.
\bjtitle{ACM Transactions on Economics and Computation (TEAC)}
\bvolume{1}(\bissue{4}),
\bfpage{1}--\blpage{26}
(\byear{2013})
\end{barticle}
\endbibitem

\bibitem[\protect\citeauthoryear{Lu et~al.}{2010}]{bib3}
\begin{bchapter}
\bauthor{\bsnm{Lu}, \binits{P.}},
\bauthor{\bsnm{Sun}, \binits{X.}},
\bauthor{\bsnm{Wang}, \binits{Y.}},
\bauthor{\bsnm{Zhu}, \binits{Z.A.}}:
\bctitle{Asymptotically optimal strategy-proof mechanisms for two-facility
  games}.
In: \bbtitle{Proceedings of the 11th ACM Conference on Electronic Commerce},
pp. \bfpage{315}--\blpage{324}
(\byear{2010})
\end{bchapter}
\endbibitem

\bibitem[\protect\citeauthoryear{Fotakis and Tzamos}{2014}]{bib4}
\begin{barticle}
\bauthor{\bsnm{Fotakis}, \binits{D.}},
\bauthor{\bsnm{Tzamos}, \binits{C.}}:
\batitle{On the power of deterministic mechanisms for facility location games}.
\bjtitle{ACM Transactions on Economics and Computation (TEAC)}
\bvolume{2}(\bissue{4}),
\bfpage{1}--\blpage{37}
(\byear{2014})
\end{barticle}
\endbibitem

\bibitem[\protect\citeauthoryear{Alon et~al.}{2009}]{bib2}
\begin{botherref}
\oauthor{\bsnm{Alon}, \binits{N.}},
\oauthor{\bsnm{Feldman}, \binits{M.}},
\oauthor{\bsnm{Procaccia}, \binits{A.D.}},
\oauthor{\bsnm{Tennenholtz}, \binits{M.}}:
Strategyproof approximation mechanisms for location on networks.
CoRR,abs/0907.2049
(2009)
\end{botherref}
\endbibitem

\bibitem[\protect\citeauthoryear{Cheng et~al.}{2013}]{bib21}
\begin{barticle}
\bauthor{\bsnm{Cheng}, \binits{Y.}},
\bauthor{\bsnm{Yu}, \binits{W.}},
\bauthor{\bsnm{Zhang}, \binits{G.}}:
\batitle{Strategy-proof approximation mechanisms for an obnoxious facility game
  on networks}.
\bjtitle{Theoretical Computer Science}
\bvolume{497},
\bfpage{154}--\blpage{163}
(\byear{2013})
\end{barticle}
\endbibitem

\bibitem[\protect\citeauthoryear{Ye et~al.}{2015}]{bib20}
\begin{bchapter}
\bauthor{\bsnm{Ye}, \binits{D.}},
\bauthor{\bsnm{Mei}, \binits{L.}},
\bauthor{\bsnm{Zhang}, \binits{Y.}}:
\bctitle{Strategy-proof mechanism for obnoxious facility location on a line}.
In: \bbtitle{International Computing and Combinatorics Conference},
pp. \bfpage{45}--\blpage{56}.
\bpublisher{Springer},
\blocation{Beijing}
(\byear{2015})
\end{bchapter}
\endbibitem

\bibitem[\protect\citeauthoryear{Serafino and Ventre}{2015}]{bib18}
\begin{bchapter}
\bauthor{\bsnm{Serafino}, \binits{P.}},
\bauthor{\bsnm{Ventre}, \binits{C.}}:
\bctitle{Truthful mechanisms without money for non-utilitarian heterogeneous
  facility location}.
In: \bbtitle{Proceedings of the AAAI Conference on Artificial Intelligence},
vol. \bseriesno{29}
(\byear{2015})
\end{bchapter}
\endbibitem

\bibitem[\protect\citeauthoryear{Serafino and Ventre}{2016}]{bib5}
\begin{barticle}
\bauthor{\bsnm{Serafino}, \binits{P.}},
\bauthor{\bsnm{Ventre}, \binits{C.}}:
\batitle{Heterogeneous facility location without money}.
\bjtitle{Theoretical Computer Science}
\bvolume{636},
\bfpage{27}--\blpage{46}
(\byear{2016})
\end{barticle}
\endbibitem

\bibitem[\protect\citeauthoryear{Yuan et~al.}{2016}]{bib6}
\begin{bchapter}
\bauthor{\bsnm{Yuan}, \binits{H.}},
\bauthor{\bsnm{Wang}, \binits{K.}},
\bauthor{\bsnm{Fong}, \binits{K.C.}},
\bauthor{\bsnm{Zhang}, \binits{Y.}},
\bauthor{\bsnm{Li}, \binits{M.}}:
\bctitle{Facility location games with optional preference}.
In: \bbtitle{ECAI 2016},
pp. \bfpage{1520}--\blpage{1527}.
\bpublisher{IOS Press},
\blocation{Amsterdam}
(\byear{2016})
\end{bchapter}
\endbibitem

\bibitem[\protect\citeauthoryear{Li et~al.}{2021}]{bib7}
\begin{bchapter}
\bauthor{\bsnm{Li}, \binits{M.}},
\bauthor{\bsnm{Lu}, \binits{P.}},
\bauthor{\bsnm{Yao}, \binits{Y.}},
\bauthor{\bsnm{Zhang}, \binits{J.}}:
\bctitle{Strategyproof mechanism for two heterogeneous facilities with constant
  approximation ratio}.
In: \bbtitle{Proceedings of the Twenty-Ninth International Conference on
  International Joint Conferences on Artificial Intelligence},
pp. \bfpage{238}--\blpage{245}
(\byear{2021})
\end{bchapter}
\endbibitem

\bibitem[\protect\citeauthoryear{Deligkas et~al.}{2023}]{bib8}
\begin{barticle}
\bauthor{\bsnm{Deligkas}, \binits{A.}},
\bauthor{\bsnm{Filos-Ratsikas}, \binits{A.}},
\bauthor{\bsnm{Voudouris}, \binits{A.A.}}:
\batitle{Heterogeneous facility location with limited resources}.
\bjtitle{Games and Economic Behavior}
\bvolume{139},
\bfpage{200}--\blpage{215}
(\byear{2023})
\end{barticle}
\endbibitem

\bibitem[\protect\citeauthoryear{Fong et~al.}{2018}]{bib9}
\begin{bchapter}
\bauthor{\bsnm{Fong}, \binits{C.K.K.}},
\bauthor{\bsnm{Li}, \binits{M.}},
\bauthor{\bsnm{Lu}, \binits{P.}},
\bauthor{\bsnm{Todo}, \binits{T.}},
\bauthor{\bsnm{Yokoo}, \binits{M.}}:
\bctitle{Facility location games with fractional preferences}.
In: \bbtitle{Proceedings of the AAAI Conference on Artificial Intelligence},
vol. \bseriesno{32}
(\byear{2018})
\end{bchapter}
\endbibitem

\bibitem[\protect\citeauthoryear{Feldman and Wilf}{2013}]{bib14}
\begin{bchapter}
\bauthor{\bsnm{Feldman}, \binits{M.}},
\bauthor{\bsnm{Wilf}, \binits{Y.}}:
\bctitle{Strategyproof facility location and the least squares objective}.
In: \bbtitle{Proceedings of the Fourteenth ACM Conference on Electronic
  Commerce},
pp. \bfpage{873}--\blpage{890}
(\byear{2013})
\end{bchapter}
\endbibitem

\bibitem[\protect\citeauthoryear{Fotakis and Tzamos}{2013}]{bib15}
\begin{bchapter}
\bauthor{\bsnm{Fotakis}, \binits{D.}},
\bauthor{\bsnm{Tzamos}, \binits{C.}}:
\bctitle{Strategyproof facility location for concave cost functions}.
In: \bbtitle{Proceedings of the Fourteenth ACM Conference on Electronic
  Commerce},
pp. \bfpage{435}--\blpage{452}
(\byear{2013})
\end{bchapter}
\endbibitem

\bibitem[\protect\citeauthoryear{Chen et~al.}{2021}]{bib12}
\begin{barticle}
\bauthor{\bsnm{Chen}, \binits{X.}},
\bauthor{\bsnm{Hu}, \binits{X.}},
\bauthor{\bsnm{Tang}, \binits{Z.}},
\bauthor{\bsnm{Wang}, \binits{C.}}:
\batitle{Tight efficiency lower bounds for strategy-proof mechanisms in
  two-opposite-facility location game}.
\bjtitle{Information Processing Letters}
\bvolume{168},
\bfpage{106098}
(\byear{2021})
\end{barticle}
\endbibitem

\bibitem[\protect\citeauthoryear{Tang et~al.}{2020}]{bib10}
\begin{bchapter}
\bauthor{\bsnm{Tang}, \binits{Z.}},
\bauthor{\bsnm{Wang}, \binits{C.}},
\bauthor{\bsnm{Zhang}, \binits{M.}},
\bauthor{\bsnm{Zhao}, \binits{Y.}}:
\bctitle{Mechanism design for facility location games with candidate
  locations}.
In: \bbtitle{International Conference on Combinatorial Optimization and
  Applications},
pp. \bfpage{440}--\blpage{452}.
\bpublisher{Springer},
\blocation{Dallas, Texas}
(\byear{2020})
\end{bchapter}
\endbibitem

\bibitem[\protect\citeauthoryear{Xu et~al.}{2021}]{bib13}
\begin{barticle}
\bauthor{\bsnm{Xu}, \binits{X.}},
\bauthor{\bsnm{Li}, \binits{B.}},
\bauthor{\bsnm{Li}, \binits{M.}},
\bauthor{\bsnm{Duan}, \binits{L.}}:
\batitle{Two-facility location games with minimum distance requirement}.
\bjtitle{Journal of Artificial Intelligence Research}
\bvolume{70},
\bfpage{719}--\blpage{756}
(\byear{2021})
\end{barticle}
\endbibitem

\bibitem[\protect\citeauthoryear{Zou and Li}{2015}]{bib11}
\begin{bchapter}
\bauthor{\bsnm{Zou}, \binits{S.}},
\bauthor{\bsnm{Li}, \binits{M.}}:
\bctitle{Facility location games with dual preference}.
In: \bbtitle{Proceedings of the 2015 International Conference on Autonomous
  Agents and Multiagent Systems},
pp. \bfpage{615}--\blpage{623}
(\byear{2015})
\end{bchapter}
\endbibitem

\bibitem[\protect\citeauthoryear{Chan et~al.}{2021}]{bib22}
\begin{bchapter}
\bauthor{\bsnm{Chan}, \binits{H.}},
\bauthor{\bsnm{Filos-Ratsikas}, \binits{A.}},
\bauthor{\bsnm{Li}, \binits{B.}},
\bauthor{\bsnm{Li}, \binits{M.}},
\bauthor{\bsnm{Wang}, \binits{C.}}:
\bctitle{Mechanism design for facility location problems: A survey}.
In: \bbtitle{30th International Joint Conference on Artificial Intelligence},
pp. \bfpage{4356}--\blpage{4365}
(\byear{2021}).
\bcomment{International Joint Conferences on Artificial Intelligence
  Organization}
\end{bchapter}
\endbibitem

\bibitem[\protect\citeauthoryear{Chan and Wang}{2023}]{bib16}
\begin{bchapter}
\bauthor{\bsnm{Chan}, \binits{H.}},
\bauthor{\bsnm{Wang}, \binits{C.}}:
\bctitle{Mechanism design for improving accessibility to public facilities}.
In: \bbtitle{Proceedings of the 2023 International Conference on Autonomous
  Agents and Multiagent Systems},
pp. \bfpage{2116}--\blpage{2124}
(\byear{2023})
\end{bchapter}
\endbibitem

\bibitem[\protect\citeauthoryear{Chan et~al.}{2024}]{bib17}
\begin{bchapter}
\bauthor{\bsnm{Chan}, \binits{H.}},
\bauthor{\bsnm{Fu}, \binits{X.}},
\bauthor{\bsnm{Li}, \binits{M.}},
\bauthor{\bsnm{Wang}, \binits{C.}}:
\bctitle{Mechanism design for reducing agent distances to prelocated
  facilities}.
In: \bbtitle{Proceedings of the 23rd International Conference on Autonomous
  Agents and Multiagent Systems},
pp. \bfpage{2180}--\blpage{2182}
(\byear{2024})
\end{bchapter}
\endbibitem

\bibitem[\protect\citeauthoryear{Qin et~al.}{2024}]{bib19}
\begin{bchapter}
\bauthor{\bsnm{Qin}, \binits{Z.}},
\bauthor{\bsnm{Chan}, \binits{H.}},
\bauthor{\bsnm{Wang}, \binits{C.}},
\bauthor{\bsnm{Zhang}, \binits{Y.}}:
\bctitle{Mechanism design for building optimal bridges between regions}.
In: \bbtitle{Annual Conference on Theory and Applications of Models of
  Computation},
pp. \bfpage{332}--\blpage{343}.
\bpublisher{Springer},
\blocation{Singapore}
(\byear{2024})
\end{bchapter}
\endbibitem

\end{thebibliography}
\end{document}